\documentclass[11pt]{article}
\usepackage{setspace}
\usepackage[utf8]{inputenc}
\usepackage{amsmath}
\usepackage{amsthm}
\usepackage{amsfonts}
\usepackage{graphicx}
\usepackage{verbatim}
\usepackage{authblk}
\usepackage[normalem]{ulem}
\usepackage{booktabs}

\usepackage[dvipsnames]{xcolor}
 \usepackage{mathptmx}     
\usepackage[titletoc,title]{appendix}
\usepackage[numbers]{natbib}
\usepackage[T1]{fontenc}
\usepackage[all]{xy}
\usepackage{enumerate}
\usepackage{comment}
\usepackage[british]{babel}
\usepackage{a4wide}
\usepackage[sans]{dsfont}
\usepackage[mathscr]{euscript}
\usepackage{mathrsfs}
\usepackage{latexsym}

\newtheorem{theorem}{Theorem}[section]

\newtheorem{proposition}{Proposition}[section]
\theoremstyle{definition}
\newtheorem{definition}{Definition}[section]
\theoremstyle{remark}
\newtheorem{remark}{Remark}[section]
\theoremstyle{remark}
\newtheorem{example}{Example}[section]
\theoremstyle{remark}
\newtheorem{assumption}{Assumption}[section]

\numberwithin{equation}{section} \makeatletter

\usepackage{array}
\usepackage{xcolor}
\usepackage{url}
\usepackage{hyperref}
\usepackage{graphicx}
\usepackage{enumerate} 
\usepackage{empheq}
\usepackage{bbm}
\usepackage[margin=1in]{geometry}

	\newcommand{\rmd}{\mathrm{d}}

\begin{document}

\title{Three options are roughly enough \\--\\
Variance-optimal semi-static hedging in the rough Heston model}

\author[1]{Paolo Di Tella\thanks{\texttt{paolo.di\_tella@tu-dresden.de}}}
\author[1]{Martin Keller-Ressel\thanks{\texttt{martin.keller-ressel@tu-dresden.de}}}
\author[1]{Beatrice Ongarato\thanks{\texttt{beatrice.ongarato@tu-dresden.de}}}

\affil[1]{Institute of Mathematical Stochastics, TU Dresden, Dresden, Germany}
\date{\today}

\maketitle
\begin{abstract}
  This paper studies semi-static variance-optimal hedging of variance swaps in the rough Heston model under a risk-neutral measure. 
We combine dynamic trading in the underlying asset with static positions in a finite basket of European options, with the aim of minimizing the variance of the residual hedging error. 
Exploiting the Markovian lift provided by the forward variance curve, the Fourier representations of the claims and the Laplace transform characterization of the log-price, we derive semi-explicit formulas for the optimal strategies and the hedging error.  We then study the optimal claim selection problem, investigating numerically how the choice of static instruments depends on the Hurst index.  Under standard parametrization, we discover that "three options are roughly enough", i.e. adding three static vanilla options reduces the hedging error by more than 95\% relative to the purely dynamic hedge for all considered roughness regimes $H\in \{0.1, 0.3, 0.5\}$.

\end{abstract}

\section{Introduction}

Hedging strategies are typically classified as dynamic or static. 
In a dynamic strategy, the position in the underlying asset is continuously adjusted over time to reduce the portfolio's exposure to market fluctuations. 
This flexibility comes at the cost of continuous monitoring and potentially high transaction costs. 
By contrast, a static strategy is fixed at inception and maintained until maturity. Although static hedging is operationally simpler, it generally 
does not lead to perfect hedging, even in complete markets.

Semi-static hedging combines the two approaches. 
It consists of a dynamic trading strategy in the underlying asset together with static positions in additional traded instruments, such as European options. 
This structure is especially relevant in markets with high transaction costs or where some instruments are not liquid enough for continuous trading. 
Moreover, since dynamic hedging is nested within the semi-static framework, adding static instruments reduces the hedging error relative to a purely dynamic strategy.

Furthermore, \cite{neuberger1994log} shows that the payoff of a variance swap admits an exact semi-static replication, obtained by combining a dynamic position in the underlying with a static position in a continuum of vanilla options across all strikes. The replication is model-free, provided that the price
 is a continuous martingale and that the realized variance is identified with the quadratic variation.
This result has two main implications. 
First, it illustrates that semi-static hedging can lead to exact replication for certain claims, even in incomplete markets. 
Second, it identifies vanilla options as natural instruments for hedging variance risk.

 In the literature, semi-static hedging has been studied from several perspectives. 
Some works focus on exact replication, as in \cite{neuberger1994log}, while others develop model-free approaches; see, e.g., \cite{beiglbock2013model,acciaio2017semi,acciaio2017space,nutz2023limits}. 
In this paper, we consider semi-static hedging under a variance-optimal criterion following the framework adopted by \cite{ditella2019semistatic,ditella2020semistatic}. Variance-optimality, introduced in  \cite{schweizer1984varianten,follmer1986hedging} to hedge risk in incomplete markets, identifies  the strategy that minimizes the variance of the residual hedging error. In a risk-neutral measure setting, variance-optimal strategies are obtained through the Galtchouk-Kunita-Watanabe decomposition, see \cite{kunita1967square}. 
For details on how to extend the solution under the physical measure, we refer to \cite{rheinlander2011hedging}.

The semi-static version of the variance-optimal hedging problem, developed in 
\cite{ditella2019semistatic,ditella2020semistatic}, is formulated under a risk-neutral measure and combines dynamic trading in the underlying asset with static positions in a finite basket of contingent claims. 
By working with a finite number of contingent claims, this formulation incorporates the practical constraint that static positions can only be taken in finitely many instruments, rather than in a continuum of options as in idealized exact replication results. 
For a fixed basket, the framework determines the optimal dynamic strategy and the optimal static weights, together with the size of the hedging error. 
The remaining question is how to choose the static instruments themselves. 
This gives rise to the optimal claim selection problem, namely the identification of the subset of available options of given cardinality, which yields the most effective semi-static hedge.

This paper focuses on the semi-static variance-optimal hedging of variance swaps in the rough Heston model, with particular attention to the numerical computation of the hedging strategies and to the selection of the static hedging instruments. 
Before stating our contributions in detail, we briefly review the main works in the literature that are most relevant to our analysis.

As already mentioned, our analysis builds on the semi-static variance-optimal framework developed in 
\cite{ditella2019semistatic}. 
Using Fourier techniques in the spirit of \cite{kallsen2010variance}, the authors combine the Galtchouk--Kunita--Watanabe decomposition with Fourier representations of the claims to obtain semi-explicit formulas for the optimal strategies. 
Their framework covers several factor models, with particular emphasis on affine models, where Laplace transforms are characterized through generalized Riccati equations \cite{duffie2003affine}. Their result, moreover, remains valid also beyond the semimartingale assumption.

The work \cite{ditella2020semistatic} is the closest benchmark for the present paper. 
There, the authors focus on the semi-static variance-optimal approach to the Heston model under a risk-neutral measure. 
They study the hedging of variance swaps, compute the hedging error numerically, and discuss the optimal selection of static instruments. 
Their results show how the use of a finite basket of options reduces the residual hedging error and how the composition of the basket affects the hedging performance. 
Our paper follows the same framework, extending the analysis from the classical Heston model to the rough Heston model.

Further contributions have extended the semi-static variance-optimal methodology in related directions. 
In \cite{callegaro2026semistatic}, the problem is studied in an affine Hawkes--Heston model under a risk-neutral measure. 
The authors identify semi-static variance-optimal strategies for variance swaps and compute the corresponding hedging error in a discontinuous market with self-exciting jumps. 
Recent work \cite{chatziandreou2026semi} extends the semi-static variance-optimal framework to a multi-asset setting, particularly focusing on covariance risk and optimal claim selection in the Wishart model. In \cite{chatziandreou2026pricing}, the same authors apply multi-asset semi-static variance-optimal hedging analysis to Pay-as-produced Power
Purchase Agreements. 

Rough volatility models provide a natural setting to extend the analysis on semi-static variance-optimal hedging. 
It is now well established that semimartingale volatility models do not fully capture several empirical features of financial markets, especially the dynamics of historical volatility and the shape of the implied volatility surface; see \cite{gatheral2022volatility}. 
Rough volatility models address these limitations by introducing memory effects through Volterra-type dynamics in the volatility. 
Prominent examples include the rough Bergomi model \cite{bayer2016pricing}, the rough Heston model \cite{el2019characteristic}, and the broader class of rough affine models \cite{larsson2023chapter}, possibly with jumps, see e.g. \cite{bondi2024rough}. 
In rough volatility models, the stock price remains a semimartingale, while the volatility process is typically non-Markovian and non-semimartingale.

The rough Heston model is particularly suitable for the present study. 
On the one hand, it incorporates rough volatility through a Volterra kernel, allowing the model to reproduce important empirical features of volatility. 
On the other hand, it retains an affine-type transform structure: its characteristic function is characterized by a Riccati--Volterra equation, as shown in \cite{el2019characteristic}. 
This makes the model compatible with Fourier-based methods and therefore with the semi-static variance-optimal techniques developed in \cite{ditella2019semistatic}. Moreover, the semimartingality of the model can be reconstructed through a Markovian lift, written in terms of the forward variance curve.

Semi-static variance-optimal hedging in the rough Heston model has been studied in \cite{martins2023aspects}, where semi-explicit formulas are derived under the physical measure for a general class of payoffs. Explicit computations, however, are provided only for call options. Their approach remains purely theoretical: no numerical implementation is carried out, and the magnitude of the residual hedging error is not assessed. Moreover, working under the physical measure makes the hedging formulas considerably more involved, whereas a risk-neutral setting allows for a more direct derivation. We also note that \cite{martins2023aspects} does not address the optimal claim selection problem.
More recently, \cite{wang2026variance} studies the fully dynamic variance-optimal hedging problem under a risk-neutral measure in the rough Hawkes--Heston model of \cite{bondi2024rough}. 
The authors compute the optimal hedge and numerically analyze the convergence of dynamic hedging strategies obtained from kernel approximations.

Our paper contributes to the existing literature by deriving semi-static variance-optimal hedging strategies for variance swaps
in the rough Heston model. The proposed framework combines dynamic trading in the underlying asset with
static positions in a finite basket of European options. We identify the corresponding hedging strategies and
the optimal static weights, and we verify the moment, measurability and integrability conditions required to
apply the Fourier-GKW interchange results in this non-Markovian rough volatility setting. We then develop
a numerical implementation based on a multifactor approximation of the rough kernel. This allows us to
compute the strategies, solve the optimal claim selection problem for different basket sizes, and evaluate the
resulting residual hedging errors. Finally, we analyze the sensitivity of both the dynamic and static compon
ents of the hedge with respect to the Hurst index, thereby quantifying the impact of volatility roughness on
the performance and composition of the semi-static hedge.

The paper is organized as follows. 
Section~\ref{sec:semistatic-hedging} introduces the semi-static variance-optimal hedging problem, and its solution. Then, it explains how it can be combined with Fourier techniques. Finally, the section presents the optimal selection problem. In Section~\ref{sec:rough-heston}, we introduce the rough Heston model under a risk-neutral measure, and recall the formulas for the Laplace transform of the log-price.
Section~\ref{sec:variance-swap} defines variance swaps and computes the optimal hedging strategies. 
Section~\ref{sec:numerical-results} reports the numerical results, including the explanation on the numerical methods employed for the computation of the numerical error. We also study the sensitivity of both static and dynamic optimal strategies to the Hurst index. 
Section~\ref{sec:conclusions} concludes.

\section{Semi-static variance-optimal hedging} \label{sec:semistatic-hedging}
This section provides an introduction to the semi-static variance-optimal hedging problem. 
We first present a rigorous formulation of the problem and recall the solution techniques developed in \cite{ditella2019semistatic,ditella2020semistatic}. 
In Section~\ref{sec:Fourier}, referring to \cite{ditella2019semistatic}, we then explain how Fourier representations of the contingent claims can be used to make the problem more analytically tractable. Lastly, in Section~\ref{sec:optimal-selection}, we present the optimal selection problem.

\paragraph{Notation.}
Let $(\Omega, \mathcal{F}, \mathbb P)$ be filtered a probability space, and $T>0$ a finite horizon. Since our analysis relies on Fourier representations of contingent claims, we use
complex-valued processes. Therefore, we adopt the following notation throughout the paper.
For $u\in\mathbb C^n$, $\overline u$ denotes its complex conjugate and $\operatorname{Re} u$ its real part. 
We denote by $\mathcal H_{\mathbb C}^2$ and $\mathcal H^2$ the Hilbert spaces of càdlàg, square-integrable, $\mathbb F$-adapted martingales, respectively complex-valued and real-valued. 
For $X,Y\in\mathcal H^2$, $\langle X,Y\rangle$ denotes their predictable covariation; in particular, $X$ and $Y$ are orthogonal if $X_0Y_0=0$ and $\langle X,Y\rangle=0$. 
If $Z^j=X^j+iY^j\in\mathcal H_{\mathbb C}^2$, $j=1,2$, we set
\[
    \langle Z^1,Z^2\rangle
    :=
    \bigl(\langle X^1,X^2\rangle-\langle Y^1,Y^2\rangle\bigr)
    +i\bigl(\langle X^1,Y^2\rangle+\langle Y^1,X^2\rangle\bigr).
\]
Finally, for $X\in\mathcal H_{\mathbb C}^2$, we define
\[
    L_{\mathbb C}^2(X)
    :=
    \left\{
    \vartheta \text{ predictable and complex-valued}:
    \mathbb E\left[
    \int_0^T |\vartheta_s|^2\,\mathrm d\langle X,\overline X\rangle_s
    \right]<\infty
    \right\},
\]
and denote by $L^2(X)$ the corresponding space of real-valued integrands when $X\in\mathcal H^2$.
\vspace{2em}

Let $S$ be the price process of a tradable asset and assume $S \in \mathcal{H}^2$. Consider a target claim $H^0 \in \mathcal{H}^2$, and a fixed basket of contingent claims $\pmb{H}:=(H^1,\dots,H^d)^\top$, $H^j \in \mathcal{H}^2$ for  $j = 1,\dots,d$. Identifying a semi-static variance-optimal hedge for $H^0$ consists in solving the optimization problem stated in Definition~\ref{definition semistatic variance optimal hedging}, first introduced in this form by \cite{ditella2019semistatic}.  

\begin{definition} \label{definition semistatic variance optimal hedging}
   A \emph{semi-static variance-optimal hedge} for the square-integrable claim $H^0$ is a triplet
$(c^\star,\vartheta^\star,\pmb{\nu}^\star)$, where $c^\star\in\mathbb{R}$ is the initial capital, 
$\vartheta^\star\in L^2(S)$ is the dynamic strategy in the underlying asset, and 
$\pmb{\nu}^\star\in\mathbb{R}^d$ denotes the static positions in the available claims, 
which attains the minimum
\begin{equation}\label{eq:semistatic-hedging-problem}
    \varepsilon^2
    =
    \min_{c \in \mathbb{R}, \vartheta \in L^2(S), \pmb{\nu}\in \mathbb{R}^d}
    \mathbb{E}
    \left[
        \left(
        c-\mathbb{E}[\pmb{\nu}^\top \pmb{H}_T]
        +\int_0^T \vartheta_s\, \mathrm{d}S_s
        -
        \bigl(H_T^0-\pmb{\nu}^\top \pmb{H}_T\bigr)
        \right)^2
    \right].
\end{equation}
\end{definition}

The approach of \cite{ditella2019semistatic,ditella2020semistatic} shows that 
problem~\eqref{eq:semistatic-hedging-problem} separates into: an inner dynamic 
variance-optimal hedging problem for the residual claim 
$H^0-\pmb{\nu}^{\top}\pmb{H}$, for fixed 
$\pmb{\nu}$; an outer finite-dimensional minimization over the static positions:
\begin{align}&\varepsilon^2(\pmb{\nu}) = \min_{c\in \mathbb{R}, \vartheta \in L^2(S)} \mathbb{E} \left[\left(c-\mathbb{E}[\pmb{\nu}^\top \pmb{H}_T]+\int_0^T \vartheta_s\, \rmd S_s -(H^0_T-\pmb{\nu}^\top \pmb{H}_T)\right)^2\right] \label{eq:inner-problem}\\
&\varepsilon^2=\min_{\pmb{\nu}\in \mathbb{R}^d} \varepsilon^2(\pmb{\nu}). \label{eq:outer-problem}
\end{align}

We refer to \eqref{eq:inner-problem} and \eqref{eq:outer-problem} as the inner and outer problems, respectively. Being a fully dynamic variance-optimal hedging problem, the solution of the inner problem is obtained through the 
Galtchouk--Kunita--Watanabe decomposition, introduced in \cite{kunita1967square}, which we recall below.

\begin{definition}[Galtchouk--Kunita--Watanabe decomposition]\label{def:GKW}
Let $H,S\in\mathcal{H}_{\mathbb{C}}^2$. 
The Galtchouk--Kunita--Watanabe decomposition of $H$ with respect to $S$ is the representation
\[
    H
    =
    H_0+\int_0^\cdot \vartheta_s\,\mathrm dS_s+L,
\]
where $\vartheta\in L_{\mathbb{C}}^2(S)$ and $L\in\mathcal{H}_{\mathbb{C}}^2$ is a martingale with $L_0=0$ such that $\langle S,\overline{L}\rangle=0.$
The pair $(\vartheta,L)$ is unique. 
The process $\vartheta$ represents the hedgeable component obtained by trading in $S$, while $L$ is the residual, non-hedgeable risk. 
Note that the integrand is given by
 $\vartheta
    =
    \frac{\mathrm d\langle H,\overline{S}\rangle}
         {\mathrm d\langle S,\overline{S}\rangle}.$
\end{definition}
By the standing assumption $H^j\in\mathcal{H}^2$, $j=0,\dots,d$, the GKW decomposition of each $H^j$ with respect to $S$ is well defined. 
Hence, for a fixed static position 
$\pmb{\nu}=(\nu_1,\nu_2,\dots,\nu_d)$, the solution of \eqref{eq:inner-problem} is the pair 
$(\vartheta^{\pmb{\nu},*},c^*)$ given by
\[
    \vartheta^{\pmb{\nu},*}
    =
    \vartheta^0-\sum_{j=1}^d \nu_j \vartheta^j,
    \qquad
    c^*=\mathbb{E}[H^0].
\]
Here, $(\vartheta^j,L^j)$ denotes the GKW decomposition of $H^j$ with respect to $S$, for $j=0,\dots,d$.
 By \cite[Theorem 2.3]{ditella2020semistatic}, once the inner problem has been solved for a fixed static position 
$\pmb{\nu}$, the corresponding hedging error is a quadratic function of $\pmb{\nu}$:
\begin{equation*}\label{eq:error-nu}
    \varepsilon^2(\pmb{\nu})
    =
    A-2\pmb{\nu}^\top B+\pmb{\nu}^\top C\pmb{\nu}.
\end{equation*}
Here $A\in\mathbb{R}$, $B\in\mathbb{R}^d$, and $C\in\mathbb{R}^{d\times d}$ are determined by the predictable covariations of the residual martingales:
\begin{equation}\label{eq:ABC}
    A:=\mathbb{E}\!\left[\langle L^0,L^0\rangle_T\right],
    \qquad
    B^j:=\mathbb{E}\!\left[\langle L^0,L^j\rangle_T\right],
    \qquad
    C^{ij}:=\mathbb{E}\!\left[\langle L^i,L^j\rangle_T\right],
    \quad i,j=1,\dots,d.
\end{equation}
Thus, the outer problem consists in minimizing the quadratic function 
$\varepsilon^2(\pmb{\nu})$ with respect to the static position 
$\boldsymbol{\nu}\in\mathbb{R}^d$. If $C$ is invertible, then the optimal static strategy is given by $\pmb{\nu}^*=C^{-1}B$.  Consequently, the semi-static variance-optimal hedge for 
problem~\eqref{eq:semistatic-hedging-problem} is given by
\begin{equation*}\label{eq:optimals}
    \pmb{\nu}^\star=C^{-1}B,
    \qquad
    \vartheta^\star
    =
    \vartheta^0-\sum_{j=1}^d \nu_j^\star\vartheta^j,
    \qquad
    c^\star=\mathbb{E}[H_T^0].
\end{equation*}
The corresponding minimal squared hedging error is
\begin{equation*}\label{eq:hedging-error}
    \varepsilon^2
    =
    \varepsilon^2(\pmb{\nu}^\star)
    =
    A-B^\top C^{-1}B.
\end{equation*}

\begin{remark}
As recalled in \cite{ditella2019semistatic}, if $C$ is not invertible, the optimal static positions are the solutions of the linear system $C\pmb{\nu}=B$.
Among them, the solution which minimizes the Euclidean norm is 
$\pmb{\nu}^\star=C^\dagger B$, where $C^\dagger$ denotes the Moore--Penrose pseudo-inverse of $C$. 
\end{remark}

\subsection{Fourier representation} \label{sec:Fourier}

Variance-optimal hedging can often be made tractable for Fourier-representable payoffs 
by combining the transform representation of the payoff with the conditional transform 
of the underlying model; see, for instance, \cite{kallsen2010variance}. 
This approach is particularly meaningful in models admitting semi-explicit formula for the Laplace transform. 
Affine models provide a classical example, since their characteristic functions can be written in terms of  generalized Riccati equations, see \cite{duffie2003affine}. 
A similar transform structure is available in rough affine models, possibly with jumps, 
where the corresponding transforms are described by Riccati--Volterra equations, see \cite{larsson2023chapter, bondi2024rough}. 
We refer to \cite{ditella2019semistatic} for an overview on Fourier representations applied to semi-static 
variance-optimal hedging, and we summarize the main ideas below.
\vspace{1em}

Assume that $S\in\mathcal{H}^2$ is strictly positive and set $X=\log S$. 
Let $H_T=h(X_T)$ be a square-integrable European payoff, $h:\mathbb{R}\to\mathbb{R}$. 
Suppose that the two-sided Laplace transform of $h$, denoted by $\widehat h$, is well defined and integrable along the strip $\mathcal{S}(R):=\{u\in\mathbb{C}:\operatorname{Re} u =R\}$,
for some $R\in\mathbb{R}$. 
Then $h$ admits the Fourier--Laplace representation
\begin{equation}\label{eq:fourier-repr}
    h(z)
    =
    \frac{1}{2\pi i}
    \int_{R-i\infty}^{R+i\infty}
        e^{uz}\widehat h(u)\,\mathrm du
    =
    \int_{\mathcal{S}(R)} e^{uz}\,\zeta(\rmd u),
\end{equation}
where $\zeta$ is the finite complex measure on $\mathcal{S}(R)$ defined by
\begin{equation}\label{eq:measure-repr}
    \zeta(\rmd u)
    :=
    \frac{1}{2\pi i}\widehat h(u)\,\rmd u.
\end{equation}

If, in addition, $\mathbb{E}[e^{R X_T}]<\infty$, the risk-neutral price of the claim at time $t$ can be written as
\begin{equation}\label{eq:fourier-repr-claim}
    H_t
    =
    \int_{\mathcal{S}(R)}
        \widetilde H_t^u\,\zeta(\mathrm du),
    \qquad t\in[0,T],
\end{equation}
with
\begin{equation*}\label{eq:conditional-laplace}
    \widetilde H_t^u
    :=
    \mathbb{E}\!\left[e^{uX_T}\mid\mathcal{F}_t\right].
\end{equation*}

\begin{remark}\label{rem:call-put-fourier}
For a European call option with strike $K>0$, corresponding to 
$h(z)=(e^z-K)^+$, the transform is
\begin{equation*}\label{eq:h-hat-call}
    \widehat h(u)
    =
    \frac{K^{1-u}}{u(u-1)},
\end{equation*}
for $R>1$. 
The same expression applies to a European put payoff 
$h(z)=(K-e^z)^+$, but with $R<0$. 
We refer to \cite[Section 4]{hubalek2006variance} for further details.
\end{remark}

The main result in \cite[Theorem 4.1, Theorem 4.2]{ditella2019semistatic} states that, under suitable conditions, the Fourier representation of the claim in equation~\eqref{eq:fourier-repr-claim}  and the GKW decomposition of the claim with respect to $S$ can be interchanged. 

\subsection{Optimal selection}\label{sec:optimal-selection}

In \cite{ditella2020semistatic}, the authors consider the problem of optimally selecting contingent claims in the Heston model. This problem is naturally connected to semi-static variance-optimal hedging: while \eqref{eq:semistatic-hedging-problem} determines the optimal semi-static strategy for a fixed set of available claims, one may further ask how many claims are needed to achieve a prescribed reduction of the hedging error, and which claims should be selected. As in \cite{ditella2020semistatic}, this leads to the constrained problem
\begin{equation*}
\begin{aligned} 
\varepsilon^2
    &=
    \min_{\pmb{\nu}\in \mathbb{R}^d} 
    A-2\pmb{\nu}^\top B+\pmb{\nu}^\top C\pmb{\nu} \\
    \|\pmb{\nu}\|_0 
    &:=
    \# \left\{ i \in \{1,\dots,d\} : \nu_i \neq 0 \right\}
    \leq m,
\end{aligned}
\label{eq:optimal-selection-problem}
\end{equation*}
where the number of non-zero static positions is required to be at most $m$. This sparsity constraint only modifies the outer optimization problem, namely the choice of the static instruments. A direct way to solve this problem is by brute-force search, namely by evaluating the hedging error for every admissible subset of cardinality $m$ and then selecting the one with the smallest error. This procedure is exact, but its computational cost grows combinatorially with $d$, since it requires the evaluation of $\binom{d}{m}$ subsets for each fixed $m$.

For large baskets, brute force search becomes infeasible. In this case, sparse selection methods such as LASSO, greedy forward selection, and Leaps-and-Bounds are possible alternatives. Leaps-and-Bounds solves the subset selection problem exactly, but its runtime may be sensitive to the model parameters. LASSO is computationally efficient, but can be unstable when the covariance matrix $C$ is ill-conditioned, especially in the presence of highly correlated claims. Greedy forward selection is only an approximation, since it constructs the subset sequentially and does not guarantee global optimality. However, it is of simple implementation, fast, and typically provides hedging errors close to the exact optimum. It therefore offers a good trade-off between implementation cost and accuracy. Refer to \cite{ditella2019semistatic} for a comprehensive comparison. In this paper, we consider a basket of $d=21$ contingent claims. For this dimension, brute-force search is still computationally feasible.

\section{The rough Heston model} \label{sec:rough-heston}

The rough Heston model, introduced in \cite{el2019characteristic}, is a rough-volatility extension of the classical Heston model. 
On the fixed time horizon $[0,T]$, under a risk-neutral measure, we set $X=\log S$. 
The log-price process $X$ and the volatility process $V$ are weak solutions, unique in law,
to the following Volterra stochastic system:
\begin{subequations}\label{system:model}
\begin{empheq}[left=\empheqlbrace]{alignat=2}  
    \rmd X_t 
    &= -\frac{1}{2}V_t\,\rmd t+\sqrt{V_t}\,\rmd W_t, 
    &\qquad X_0\in\mathbb R, 
    \label{eq:model-X}\\
    V_t 
    &= V_0+\int_0^t K(t-s)
    \left(\lambda(\theta-V_s)\,\rmd s+\sigma\sqrt{V_s}\,\rmd B_s\right),
    \label{eq:model-V}
\end{empheq}
\end{subequations}
for $t\in[0,T]$. 
Here $W$ and $B$ are Brownian motions with constant instantaneous correlation 
$\rho\in(-1,1)$, that is, $\rmd\langle W,B\rangle_t=\rho\,\rmd t.$
The parameters $X_0\in\mathbb R$, $V_0>0$, $\lambda>0$, $\theta>0$, $\sigma>0$, are constants. 
The kernel $K:(0,\infty)\to\mathbb R_+$ is assumed to satisfy the following standard conditions; see \cite[Equation (5.6)]{larsson2023chapter}:
$K$ is strictly positive and completely monotone. 
Moreover, there exists $\gamma\in(0,2]$ such that
\begin{equation*} \label{eq: gamma}
    \int_0^h K(t)^2\,\rmd t = O(h^\gamma),
    \qquad
    \int_0^T \bigl(K(t+h)-K(t)\bigr)^2\,\rmd t = O(h^\gamma),
\end{equation*}
for every $T<\infty$.

If the kernel satisfies the conditions above, then weak existence and uniqueness in law of the solution to 
\eqref{system:model}, with $V_t\geq 0$ for all $t\in[0,T]$ a.s., follow from \cite[Theorems 3.6 and 7.1]{abi2019affine}. 

Recalling $S=e^X$, the dynamics of $S$ is given by:

\begin{equation} \label{eq:price-S}
  \rmd S_t=S_t\sqrt{V_t}\,\rmd W_t, \qquad S_0=e^{X_0}.  
\end{equation}

In general, the rough volatility process $V$ is neither Markov nor a semimartingale. The Markov property can be recovered by introducing the so-called \emph{forward variance curve}, which can be interpreted as an infinite-dimensional Markovian lift of the volatility process, in the spirit of \cite{cuchiero2020generalized}. We define the forward variance curve as $(\xi_t(u))_{0\leq t \leq u \leq T}$
\begin{equation} \label{forward var}
    \xi_t(u)=\mathbb{E}[V_u\vert \mathcal{F}_t].
\end{equation}

{\color{black} We recall the dynamics of $(\xi_t(u))_{0\leq t \leq u \leq T}$. To this aim, we introduce the definition of resolvent of the second kind of the kernel, referring to \cite[Chapter 2, Theorem 3.1]{gripenberg1990volterra}.}
    
\begin{definition}
    For $f, g \in L^1_{\text{loc}}(\mathbb{R}_+; \mathbb{C})$, we define the convolution $f*g\in L^1_{\text{loc}}(\mathbb{R}_+; \mathbb{C})$ by
\[(f*g)(t):=\int_0^t f(t-s) g(s) \rmd s = \int_0^t f(s) g(t-s) \rmd s, \quad t\geq 0.\]
\end{definition}
\begin{definition}
 Given $K\in L^2_{\text{loc}}(\mathbb{R_+})$, we define $R\in L^2_{\text{loc}}(\mathbb{R_+})$, the resolvent of the second kind of $K$, as the unique solution to
    \[K(t)-R(t)= (R * K) (t).\]
\end{definition}
\begin{example} \label{fractional kernel}
In the fractional kernel case
\begin{align*}
    K(t)=c\frac{t^{H-\frac{1}{2}}}{\Gamma(H+1/2)}, \qquad R(t)=c t^{H-\frac{1}{2}} E_{H+1/2, H+1/2}(-ct^{H+\frac{1}{2}}), \qquad H\in (0, 1/2)
\end{align*}
where $H$ denotes the Hurst index, and  $E_{\alpha,\beta}(z)=\sum_{n=0}^\infty \frac{z^n}{\Gamma(\alpha n +\beta)}$ 
the Mittag–Leffler function, \cite{wiman1905fundamentalsatz}. Other examples of the resolvent of the second kind for kernels which satisfy standard assumptions are provided in \cite[Table 1]{abi2019affine}. Note that the classical Heston case can be reconstructed by taking a constant kernel, for which the resolvent of the second kind becomes the exponential function.
\end{example}
 The following proposition identifies the dynamics for the forward variance curve as studied in \cite[Lemma I.2.10]{martins2023aspects} and \cite[Proposition 5.4]{larsson2023chapter}.  Let $R_\lambda$ be the resolvent of the second kind of $\lambda K$, we define $\widehat K:=\frac{\sigma}{\lambda} R_{\lambda}.$
\begin{proposition}
The forward variance $(\xi_t(T))_{0\leq t\leq T}$ satisfies
    \begin{equation*}\label{forward variance sde}\rmd\xi_t(u)=\widehat K (u-t)\sqrt{V_t} \rmd B_t, \quad 0\leq t \leq u,\end{equation*}
    with initial condition\[\xi_0(u)=V_0 + \frac{\lambda(\theta-V_0)}{\sigma}\int_0^u \widehat K (s) \rmd s, \quad \xi_u(u)=V_u,\]
    for every $u\in [0,T]$.
\end{proposition}

Moreover, since the rough Heston model is a rough affine model, we recall the formulas for its Laplace transform, see \cite[Example 7.2]{abi2019affine}, \cite[Theorem 5.5]{larsson2023chapter}, \cite[Corollary I.2.19]{martins2023aspects}.

\begin{definition} \label{explosion time}
    Let $T^*(a):=\sup\{ t\geq 0: \mathbb{E}[e^{a X_t}]<\infty\}$, $a\in \mathbb{R}$, be the moment explosion time of rough Heston. 
\end{definition}

\begin{proposition} \label{conditional expectation}
Assume $u\in \mathbb{C}$, and $T<T^*(\operatorname{Re}u)$, then for every $t\in [0,T]$, it holds
\begin{align}
    \widetilde H^u_t = \mathbb{E}[e^{uX_T}\vert \mathcal{F}_t] = \exp\left( u X_t + \int_t^T g_{u}(T-s)\xi_t(s) \rmd s\right) \label{eq:H-tilde}
\end{align}
    where $g_{u} \in L^1([0,T]; \mathbb{C})$ is the unique solution to the equation
    \begin{align} \label{g 1u}
       g_{u}(t)=\frac{1}{2}(u^2-u) + \rho u (\widehat K * g_{u})(t)+\frac{1}{2} \left( (\widehat K * g_{u})(t)\right)^2
    \end{align}

for $0\leq t \leq T$. Moreover, $g_{u}$ is continuous in $[0,T]$.
\end{proposition}

We also recall the following property, introduced in \cite[Corollary I.2.14]{martins2023aspects}. 

\begin{proposition} \label{dynamics int xi}
For any function $f\in L^1([0,T]; \mathbb{C})$, there exists a continuous semimartingale $(V_t^f)_{0\leq t \leq T}$ such that
\[V_t^f = \int_t^T f(T-u)\xi_t(u)\rmd u \quad 0\leq t \leq T,\]
in particular $V_T^f=0$, and can be written as
\[V_t^f = V_0^f -\int_0^t f(T-s) V_s \rmd s + \int_0^t (\widehat K * f)(T-s) \sqrt V_s \rmd B_s, \quad 0\leq t \leq T.  {\color{red}}\]
\end{proposition}

\section{Hedging a variance swap}\label{sec:variance-swap}
This section focuses on the computation of the optimal semi-static variance-optimal hedge for a variance swap in the rough Heston model, where the static component consists of European claims with Fourier-representable payoffs, such as calls and puts.

Let $H^0$ be a variance swap written on the stock $S=e^X$. 
More precisely, the payoff at maturity $T >0$ of the variance swap is defined as
\begin{equation*} \label{eq: def H_T^0}
    H^0_T= [X,X]_T-k,
\end{equation*}
where $k$ is the so-called swap rate, i.e., $k=\mathbb{E}[[X,X]_T]$ so that $\mathbb{E}[H^0_T]=0$ and the contract is zero at inception.
When $X$ is given by \eqref{eq:model-X}, then

\[H_T^0=\int_0^T V_s \rmd s-k,\]

and
\begin{equation}\label{eq:variance-swap-t}
    H_t^0:=\mathbb{E}[H_T^0\vert \mathcal{F}_t] = \int_0^t V_s \rmd s + \int_t^T \mathbb{E}[V_s \vert \mathcal{F}_t] \rmd s -k = \int_0^t V_s \rmd s + \int_t^T \xi_t(s) \rmd s -k.
\end{equation} 
The set of contingent claims used to  hedge is a basket $\pmb{H}=(H^1, \dots, H^d)^\top$ of European claims written on $S=e^X$.

\begin{assumption}\label{assumption T}
    Let $R_1<0$ and $R_2>1$. Assume that $T < \min\{T^*(2R_1), T^*(2R_2)\}$.
\end{assumption}

\begin{remark}\label{rem:explosion-time-heston}
According to \cite{gerhold2019moment}, the moment explosion time for a given order $u\in \mathbb{R}$ in the rough Heston model is finite if and only if the corresponding explosion time in the classical Heston model is finite. Moreover, the explosion time of the Riccati--Volterra equation associated with $g_u$ in \eqref{g 1u} coincides with the explosion time of the moment $\mathbb{E}[e^{uX_T}]$. Since the classical Heston model admits an explicit characterization of moment explosions in terms of its parameters, see e.g. \cite{andersen2007moment}, a simple sufficient condition for Assumption~\ref{assumption T} is to require
$T^*_{\mathrm{Heston}}(2R_1)=T^*_{\mathrm{Heston}}(2R_2)=\infty.$
This implies that the corresponding rough Heston explosion times are infinite, so that Assumption~\ref{assumption T} holds for any finite maturity $T$. Alternatively, one can exploit some a-priori-bounds on the rough Heston explosion time, see \cite{gerhold2019moment} and \cite{keller2020comparison} for further details.

\end{remark}

\begin{proposition}\label{prop:uniform_explosion_interval}
If Assumption~\ref{assumption T} holds, then 
$T<\inf_{a\in[2R_1,2R_2]}T^*(a).$
\end{proposition}

\begin{proof}
By assumption, $T<\min\{T^*(2R_1), T^*(2R_2)\}$.Hence there exists $\varepsilon>0$ such that
$T+\varepsilon<\min\{T^*(2R_1), T^*(2R_2)\}.$
Therefore
\[T+\varepsilon<T^*(2R_1),
\qquad
T+\varepsilon<T^*(2R_2),
\]
i.e., referring to \cite[Theorem 7]{gerhold2019moment}
\[
\mathbb E\left[e^{2R_1X_{T+\varepsilon}}\right]<\infty,
\qquad
\mathbb E\left[e^{2R_2X_{T+\varepsilon}}\right]<\infty.
\]

Let now $a\in[2R_1,2R_2]$. Then there exists $\lambda\in[0,1]$ such that
$a=\lambda 2R_1+(1-\lambda)2R_2.$
Hence, by Hölder's inequality,
\[
\begin{aligned}
\mathbb E\left[e^{aX_{T+\varepsilon}}\right]
&=
\mathbb E\left[
e^{\lambda 2R_1X_{T+\varepsilon}}
e^{(1-\lambda)2R_2X_{T+\varepsilon}}
\right] \leq
\left(
\mathbb E\left[e^{2R_1X_{T+\varepsilon}}\right]
\right)^\lambda
\left(
\mathbb E\left[e^{2R_2X_{T+\varepsilon}}\right]
\right)^{1-\lambda}
<\infty.
\end{aligned}
\]
Therefore
$T+\varepsilon\leq T^*(a)$
for every $a\in[2R_1,2R_2]$. Consequently,
$\inf_{a\in[2R_1,2R_2]}T^*(a)\geq T+\varepsilon>T.$
\end{proof}

\begin{proposition}     \label{prop:square-int-mg}Assumption~\ref{assumption T} holds. Then $(S_t)_{0\leq t\leq T}, (H_t^0)_{0\leq t\leq T}$ are square integrable martingales. Moreover, assume $u\in \mathbb{C}$, such that $T<T^*(2\operatorname{Re}u)$, then  also $(\widetilde H_t^u)_{0\leq t\leq T}$ is a square-integrable martingale. In particular, it holds for $u\in \mathbb{C}$ such that $\operatorname{Re}(u) \in [R_1, R_2]$.
\end{proposition}

\begin{proof}
Recall $R_1<0$ and $R_2>1$, so by Proposition~\ref{prop:uniform_explosion_interval}, it implies $T<T^*(1),\, T< T^*(2)$.
Since $T < T^*(2)$, $\mathbb{E}[e^{2X_T}]=\mathbb{E}[S_T^2]<\infty$, i.e. $S$ is square-integrable. In particular,  $S$ is also a martingale, since it is integrable by hypothesis ($T<T^*(1)$, $\mathbb{E}[e^{X_T}]=\mathbb{E}[S_T]<\infty$) and applying Proposition~\ref{conditional expectation},
\begin{align*}
    \mathbb{E}[S_T\vert\mathcal{F}_t] = \mathbb{E}[e^{X_T}\vert\mathcal{F}_t] = \exp\left(X_t+\int_t^T g_{1}(T-s) \xi_t(s)\right),
\end{align*}
where $g_1$ is the unique solution to \eqref{g 1u} evaluated in $u=1$, i.e.
\[g_1(t)=\rho (\widehat K*g_1)(t)+\frac{1}{2}((\widehat K*g_1)(t))^2.\]
By substitution, it follows that $g_{1}(t)\equiv 0$ is the unique solution, thus
$\mathbb{E}[S_T\vert\mathcal{F}_t] = \mathbb{E}[e^{X_T}\vert\mathcal{F}_t] = \exp(X_t)=S_t.$
Next, recall that
\[
    H_T^0=\int_0^T V_s\,\rmd s-k,
    \qquad
    H_t^0=\mathbb E[H_T^0\mid\mathcal F_t].
\]
Thus, once $H_T^0\in L^2$ is shown, $(H_t^0)_{0\leq t\leq T}$ is a square-integrable martingale by standard martingale theory.
\begin{align*}
    \mathbb{E}[(H_T^0)^2] &= \mathbb{E}\left[ \left( \int_0^T V_s \rmd s  -k \right)^2 \right] \leq C \mathbb{E}\left[ \left( \int_0^T V_s \rmd s \right)^2 \right] +Ck^2\leq C \int_0^T \mathbb{E}[V_s^2] \rmd s +Ck^2 <\infty,
\end{align*}
for some constant $C>0$.
The last estimate follows by \cite[Lemma 3.1 and Theorem 3.4]{abi2019affine}. Applying their
moment estimate to the rough Heston variance process yields, for every fixed $T > 0$, a constant $c_T > 0$ such that
$\sup_{t\leq T} \mathbb{E}[V_t^2]\leq c_T,$
meaning $\int_0^t \mathbb{E}[V_s^2] \rmd s \leq c_T t.$ 

Finally, let $u\in\mathbb C$ satisfy $T<T^*(2\operatorname{Re}u)$. Since $\widetilde H_t^u
    =
    \mathbb E[e^{uX_T}\mid\mathcal F_t]$,
the process $\widetilde H^u$ is a martingale. Therefore, it is enough to prove $\widetilde H^u_T\in L^2$:
\begin{align*}
\mathbb{E}[\vert\mathbb{E}[\exp(uX_T)] \vert^2] \leq \mathbb{E}[\mathbb{E}[|\exp(uX_T)|^2]] =\mathbb{E}[\exp(2\operatorname{Re}uX_T)]<\infty.\qquad\qedhere
\end{align*} 
\end{proof}

Theorem~\ref{thm:strategies-rough-heston} provides the optimal semi-static variance-optimal hedge for a variance swap when a basket of European options is available for static trading.

\begin{theorem}\label{thm:strategies-rough-heston}
Suppose that Assumption~\ref{assumption T} holds. Let $R_j\in\mathbb R$,
$j=1,\dots,d$, and let $0<\varepsilon<2$ be such that
$\mathbb E[e^{2R_jX_T}]<\infty,
    \,
    \mathbb E[e^{(2+\varepsilon)R_jX_T}]<\infty,
    \, j=1,\dots,d.$
Define $\mathcal S^j:=\{u\in\mathbb C:\operatorname{Re}(u)=R_j\},\, j=1,\dots,d.$
Consider real-valued square-integrable European options
$H^j=h^j(X_T)$, $j=1,\dots,d$, where $h^j$ admits a Fourier representation as in
\eqref{eq:fourier-repr}--\eqref{eq:measure-repr} on $\mathcal S^j$ with respect to
the measure $\zeta^j$.

Then the semi-static variance-optimal hedge of the variance swap through a basket of European contingent claims are given by:
\begin{align*}
    \vartheta_t^0
    &=
    \frac{\rho(\widehat K*1)(T-t)}{S_t},\\
    \vartheta_t^j
    &=
    \frac{1}{S_t}
    \int_{\mathcal S^j}
        \widetilde H_t^u
        \left(
            u+\rho(\widehat K*g_u)(T-t)
        \right)
    \zeta^j(du),\\
    A
    &=
    \int_0^T
        (1-\rho^2)
        ((\widehat K*1)(T-t))^2
        \mathbb E[V_t]
    dt,\\
    B^j
    &=
    \int_0^T
    \int_{\mathcal S^j}
        (1-\rho^2)
        (\widehat K*g_u)(T-t)
        (\widehat K*1)(T-t)
        \mathbb E[\widetilde H_t^u V_t]
    \zeta^j(du)dt,\\
    C^{ij}
    &=
    \int_0^T
    \int_{\mathcal S^i}
    \int_{\mathcal S^j}
        (1-\rho^2)
        (\widehat K*g_{u_i})(T-t)
        (\widehat K*g_{u_j})(T-t)
        \mathbb E[\widetilde H_t^{u_i}\widetilde H_t^{u_j}V_t]
    \zeta^i(du_i)\zeta^j(du_j)dt.
\end{align*}
\end{theorem}

\begin{proof}

The employed techniques are analogous to the ones in \cite[Proposition II.2.18, Theorem II.2.29]{martins2023aspects}.  Proposition~\ref{prop:square-int-mg} ensures that $S$, $H^0$, and 
$\widetilde H^u$, for $u\in\mathcal S^j$, are square-integrable martingales. 
We therefore first compute the Galtchouk--Kunita--Watanabe decompositions of 
$H^0$ and $\widetilde H^u$ with respect to $S$, and then use the Fourier 
representation of the contingent claim to obtain the desired hedging strategy.
We first verify that the Fourier-GKW interchange results of
\cite{ditella2019semistatic} apply in the present setting. 
By \cite[Proposition 2.11]{martins2023aspects}, the forward variance curve
can be chosen with almost surely Hölder-continuous paths in $(t,s)$, of any order
strictly below $\gamma/2$. We work with this version throughout. In particular,
$\xi$ is jointly measurable.

Fix $j\in\{1,\dots,d\}$ and $u\in\mathcal S^j$. The deterministic Riccati--Volterra solution $g_u$ can be chosen Borel
measurable in $(u,t)$. Indeed, this follows from the Picard construction and
uniqueness of the Riccati--Volterra solution: starting from $g_u^{(0)}=0$, each
Picard iterate is Borel measurable in $(u,t)$, and the pointwise limit is the unique
solution $g_u$. Therefore the affine representation
\[
    \widetilde H_t^u
    =
    \exp\left(
        uX_t+\int_t^T g_u(T-s)\xi_t(s)\,ds
    \right)
\]
defines a jointly progressively measurable version of $\widetilde H^u$. 
Furthermore, for $u=R_j+iy\in\mathcal S^j$, conditional Jensen's inequality gives
\[
    \mathbb E[|\widetilde H_t^u|^2]
    =
    \mathbb E\left[
        \left|
        \mathbb E[e^{uX_T}\mid\mathcal F_t]
        \right|^2
    \right]
    \le
    \mathbb E[e^{2R_jX_T}]
    <\infty.
\]
Hence $\sup_{u\in\mathcal S^j}
    \mathbb E[|\widetilde H_t^u|^2]
    <\infty.$

Recall $H^0$ is given by equation~\eqref{eq:variance-swap-t}
    \[H^0_t = \int_0^t V_s \rmd s + \int_t^T \xi_t(s) \rmd s -k.\]
    By Proposition~\ref{dynamics int xi}, we have
    \[\rmd \left( \int_t^T \xi_t(s) \rmd s \right)=-V_t \rmd t + (\widehat K * 1)(T-t) \sqrt{V_t} \rmd B_t,\]
    thus the dynamics of $H^0$ is given by 
    \[\rmd H^0_t = (\widehat K * 1)(T-t) \sqrt{V_t} \rmd B_t.\]
    Recalling the dynamics of $S$ in equation~\eqref{eq:price-S}, and the expression for $\vartheta$ in Definition~\ref{def:GKW}, it follows
    \[\vartheta^0_t := \frac{\rmd \langle H^0, S \rangle_t }{\rmd \langle S, S \rangle_t} = \frac{\rho (\widehat K * 1)(T-t)}{S_t}.\]
    Moreover, given $L$ as in Definition~\ref{def:GKW},

    \begin{align*}\rmd L_t^0 &= \rmd H_t^0-\rmd (\vartheta^0 \cdot S)_t = (\widehat K * 1)(T-t) \sqrt{V_t} \rmd B_t - \frac{\rho (\widehat K * 1)(T-t)}{S_t} \rmd S_t\\
    &=(\widehat K * 1)(T-t) \sqrt{V_t} (\rmd B_t -\rho \rmd W_t). 
    \end{align*}
    Then
    \begin{align*}
        \frac{\rmd \langle L^0, L^0 \rangle_t}{\rmd t}=(1-\rho^2)((\widehat K * 1)(T-t))^2 V_t.
    \end{align*}

Similarly, we determine the dynamics of $\widetilde H^u_t$. We recall the expression for $\widetilde H_t^u$ in equation~\eqref{eq:H-tilde}
\[\exp\left( u X_t + \int_t^T g_{u}(T-s)\xi_t(s) \rmd s\right).\]
Exploiting Proposition~\ref{dynamics int xi}
\begin{align*}
    &\rmd \left( \int_t^T  g_{u}(T-s) \xi_t(s) \rmd s \right)= -g_{u}(T-t)V_t \rmd t + (\widehat K * g_{u})(T-t) \sqrt{V_t} \rmd B_t,
\end{align*}
and the dynamics of $X=\log S$ given by 

\[\rmd X_t = -\frac{1}{2} V_t \rmd t + \sqrt{V_t} \rmd W_t.\]

It follows that
\begin{align*}
    \rmd \widetilde H_t^u = \widetilde H_t^u(u\sqrt{V_t}\rmd W_t + (\widehat K * g_{u})(T-t) \sqrt{V_t}\rmd B_t).
\end{align*}Then 
\begin{align}
    \vartheta_t^u := \frac{\rmd \langle \widetilde H^u, S \rangle_t}{\rmd \langle S, S \rangle_t} = \frac{\widetilde H_t^u(u + \rho (\widehat K * g_{u})(T-t) )}{S_t}. \label{eq:vartheta_u}
\end{align}
Employing again the definition of $L$ given in Definition~\ref{def:GKW}, we write
\begin{align*}
    \rmd L_t^u &= \rmd \widetilde H_t^u - \rmd(\vartheta^u\cdot S)_t \\
    &= \widetilde H_t^u(u\sqrt{V_t}\rmd W_t + (\widehat K * g_{u})(T-t) \sqrt{V_t}\rmd B_t) -  \frac{\widetilde H_t^u(u + \rho (\widehat K * g_{u})(T-t) )}{S_t} \rmd S_t \\
    & = \widetilde H_t^u (\widehat K * g_{u})(T-t) \sqrt{V_t} (\rmd B_t - \rho \rmd W_t).
\end{align*}
Then, recalling the expression for $L^0$, it follows
\begin{align*}
    \frac{\rmd \langle L^0, L^u \rangle_t}{\rmd t} & = (1-\rho^2) (\widehat K * 1)(T-t) (\widehat K * g_{u})(T-t) \widetilde H_t^u V_t,
\end{align*}
and
\begin{align*}
    \frac{\rmd \langle L^{u_i}, L^{u_j} \rangle_t}{\rmd t} & = (1-\rho^2) (\widehat K * g_{u_i})(T-t) (\widehat K * g_{u_j})(T-t) \widetilde H_t^{u_i}  \widetilde H_t^{u_j} V_t.
\end{align*}

{\color{black}It remains to verify the joint predictability assumption in
\cite[Theorem 4.6]{ditella2019semistatic}. By the version fixed above,
    $(t,\omega,u)\mapsto \widetilde H_t^u(\omega)$ is jointly measurable, and for each fixed $u$,
the process $\widetilde H^u$ has continuous adapted paths by the stochastic differential
representation just derived. Hence $(t,\omega,u)\mapsto \widetilde H_t^u(\omega)$ is jointly predictable, for instance by the
standard dyadic approximation argument.

Since $S$ is continuous, adapted and strictly positive, and since
$(u,t)\mapsto(\widehat K*g_u)(T-t)$ is Borel measurable, the explicit formula for $\vartheta^u$, equation~\eqref{eq:vartheta_u},
shows that $(t,\omega,u)\mapsto\vartheta_t^u(\omega)$ is jointly predictable.}

Recalling the definition of $A, B, C$ in equation~\eqref{eq:ABC}, and the interchange between the GKW decomposition and the Fourier representation, we can write:
 \begin{align*}
    \vartheta^j_t&= \int_{\mathcal{S}^j} \vartheta^u_t\, \zeta^j(\rmd u), \\
        A&=\int_0^T\mathbb{E}\left[\frac{\rmd \langle L^0, L^0 \rangle_t}{\rmd t}\right] \rmd t,\\
        B^j&=\int_0^T \int_{\mathcal{S}^j}\mathbb{E}\left[\frac{\rmd\langle L^0, L^{u} \rangle_t}{\rmd t}\right] \zeta^j(\rmd u) \rmd t,\\
        C^{i,j}&=\int_0^T \int_{\mathcal{S}^i} \int_{\mathcal{S}^j} \mathbb{E}\left[\frac{\rmd\langle L^{u_i}, L^{u_j}\rangle_t}{\rmd t}\right] \zeta^j(\rmd u_j) \zeta^i(\rmd u_i)\rmd t,
    \end{align*}
    for $i,j=1, \dots, d.$ The statement of the theorem follows by substituting the expressions obtained for 
$\vartheta_t^u$ and for the predictable covariations.

Lastly, we verify the expectations in the statement are actually finite. The expectation of $V_t$ follows directly by the definition of the forward variance curve, equation~\eqref{forward var}. Recalling that $V$ is square integrable by \cite[Lemma 3.1 and Theorem 3.4]{abi2019affine}, we can write

\[\mathbb{E}[|\widetilde H_t^{u_j} V_t|]\leq (\mathbb{E}[|\widetilde H_t^{u_j}|^2])^{1/2}(\mathbb{E}[V_t^2])^{1/2} \leq (\mathbb{E}[e^{2\operatorname{Re}u_j X_T}])^{1/2}(\mathbb{E}[V_t^2])^{1/2}.\]

Finally, let 
$p = 2+\varepsilon, \,
r = \frac{2+\varepsilon}{\varepsilon}$.
Then $1/p + 1/p + 1/r = 1$, so by Hölder's inequality,
\[
\mathbb{E}\!\left[\, |\widetilde H_t^{u_j} \widetilde H_t^{u_i} V_t| \,\right]
    \leq (\mathbb{E}[|\widetilde H_t^{u_j}|^{2+\varepsilon}])^{\frac{1}{2+\epsilon}} (\mathbb{E}[|\widetilde H_t^{u_i}|^{2+\varepsilon}])^{\frac{1}{2+\epsilon}}
      (\mathbb{E}[V_t^{\frac{2+\epsilon}{\epsilon}}])^{\frac{\epsilon}{2+\epsilon}},
\]
where the latter is finite due to \cite[Lemma 3.1 and Theorem 3.4]{abi2019affine}, noticing that $\frac{2+\varepsilon}{\varepsilon}>2$ by hypothesis.
\end{proof}

\begin{example}\label{ex:strat-for-frac-kernel}
    In the fractional case, see Example~\ref{fractional kernel}, $\widehat K=\sigma t^{H-\frac{1}{2}} E_{H+1/2, H+1/2}(-\lambda t^{H+\frac{1}{2}})$, so 
\[(\widehat K * 1) (T-t)=\int_0^{T-t} \widehat K(s) \rmd s = \sigma (T-t)^{H+\frac{1}{2}} E_{H+1/2, H+3/2}(-\lambda (T-t)^{H+\frac{1}{2}}).\]
Moreover, recall that $\mathbb{E}[V_t]=\xi_0(t)=V_0 + \lambda(\theta-V_0)t^{H+\frac{1}{2}} E_{H+1/2, H+3/2}(-\lambda t^{H+\frac{1}{2}})$. So for the fractional case the hedging strategies become
\begin{equation}
\begin{aligned}
\vartheta^0_t & =\frac{\sigma (T-t)^{H+\frac{1}{2}} E_{H+1/2, H+3/2}(-\lambda (T-t)^{H+\frac{1}{2}})\rho}{S_{t}},\\
\vartheta^j_t&=\frac{1}{S_{t}}\int_{\mathcal{S}^j}\left( \widetilde H_t^{u} (u + \rho (\widehat K * g_{u})(T-t)\right) \zeta^j(\rmd u), \\
    A&= \int_0^T (1-\rho^2)(\sigma (T-t)^{H+\frac{1}{2}} E_{H+1/2, H+3/2}(-\lambda (T-t)^{H+\frac{1}{2}}))^2 (V_0 + \lambda(\theta-V_0)t^{H+\frac{1}{2}} E_{H+1/2, H+3/2}(-\lambda t^{H+\frac{1}{2}}))\rmd t, \\
    B^j&=\int_0^T \int_{\mathcal{S}^j} (1-\rho^2) (\widehat K * g_{u})(T-t) \sigma (T-t)^{H+\frac{1}{2}} E_{H+1/2, H+3/2}(-\lambda (T-t)^{H+\frac{1}{2}}) \mathbb{E}[\widetilde H^{u}_t V_t ] \zeta^j(\rmd u) \rmd t,\\
    C^{i,j}&=\int_0^T  \int_{\mathcal{S}^i} \int_{\mathcal{S}^j} (1-\rho^2) (\widehat K * g_{u_i})(T-t)(\widehat K * g_{u_j})(T-t) \mathbb{E}[\widetilde H^{u_i}_t \widetilde H^{u_j}_t V_t] \zeta^i(\rmd u_i) \zeta^j(\rmd u_j) \rmd t. 
\end{aligned}
\label{eq:strategies-frac-case}
\end{equation}
\end{example}

\section{Numerical results}\label{sec:numerical-results}

In this section, we first address the numerical computation of the hedging strategies. 
We then introduce the optimal claim selection problem and provide a numerical analysis of the optimal choice of contingent claims for different values of the Hurst index $H$.

\subsection{Numerical computations of the strategies}

Dynamic strategies are determined by coupling the Euler scheme for the simulation of $S$ with numerical quadrature for the integrals over the complex contour $\mathcal{S}^{j}$. Conversely, static strategies are computed using Monte Carlo simulations to estimate the expectations, combined with numerical quadrature for integration both in time and over the complex contours.

Throughout the numerical experiments, we consider the fractional kernel
$K(t)=\frac{t^{H-\frac12}}{\Gamma(H+\frac12)},\, H\in(0,1/2)$,
as in Example~\ref{fractional kernel}.

We first choose a maturity $T$ such that the assumptions of Theorem~\ref{thm:strategies-rough-heston} are satisfied. 
In particular, $S$, $H^0$, and $\widetilde H^{u_j}$ are required to be square-integrable martingales for the relevant values of $u_j$, and the expectations appearing in Theorem~\ref{thm:strategies-rough-heston} must be well defined. 
We then introduce the uniform time grid
\[
0=t_0<t_1<\cdots<t_N=T, 
\qquad 
t_k=k\Delta t, 
\qquad 
\Delta t=\frac{T}{N}.
\]

We need to simulate $(S, V,\widetilde H^{u_j})$ for each $u_j\in\mathcal S_j$. 
We first describe the simulation of the rough Heston dynamics. 
To this end, we use the multifactor approximation proposed in \cite{abi2019multifactor}. 
The main idea is to approximate the fractional kernel by a finite sum of exponentials,
\[
K^n(t)=\sum_{i=1}^n c_i^n e^{-\gamma_i^n t}.
\]
The coefficients are defined by
\begin{equation}\label{coeff multifactor approximation}
c_i^n = \int_{\eta_{i-1}^n}^{\eta_i^n} \mu(\mathrm d\gamma),
\qquad
\gamma_i^n 
= \frac{1}{c_i^n}
\int_{\eta_{i-1}^n}^{\eta_i^n} \gamma \, \mu(\mathrm d\gamma),
\end{equation}
where
$\mu(\mathrm d\gamma)
=
\frac{\gamma^{-H-\frac12}}
{\Gamma(H+\frac12)\Gamma(\frac12-H)}
\,\mathrm d\gamma.$
The partition points are given by $\eta_i^n=i\pi_n$, $i=0,\dots,n$, with\\
$\pi_n
=
\frac{n^{-1/5}}{T}
\left(
\frac{\sqrt{10}(1-2H)}{5-2H}
\right)^{2/5}.$

This approximation leads to the multifactor stochastic volatility model $(X^n,V^n)$, given by
\begin{equation} \label{path x}
\rmd X_t^n
=
-\frac12 V_t^n\,\rmd t
+
\sqrt{V_t^n}\,\rmd W_t,
\qquad
V_t^n
=
V_0
+
\lambda\theta\int_0^t K^n(t-s)\,\rmd s
+
\sum_{i=1}^n c_i^n V_t^{n,i},
\end{equation}
where
\begin{equation} \label{eq:sim-V^n}
\rmd V_t^{n,i}
=
\left(
-\gamma_i^n V_t^{n,i}
-\lambda V_t^n
\right)\rmd t
+
\sigma\sqrt{V_t^n}\,\rmd B_t,
\qquad
V_0^{n,i}=0.
\end{equation}
The convergence of this multifactor approximation to the original rough volatility model is established in \cite[Theorem~3.5]{abi2019multifactor}.

After fixing the number $M$ of Monte Carlo trajectories and the number $n$ of factors, we simulate $M$ approximate paths of $(X^n,V^n)$ on the time grid using an Euler scheme. 
Since the discretized process $V^n$ in \eqref{eq:sim-V^n} may become negative, we follow \cite[Appendix~A.4]{abi2019lifting} and replace $\sqrt{\cdot}$ by $\sqrt{\max(0,\cdot)}$. For the dynamic strategies we use $S^n=e^{X^n}$.

We now describe the computation of $\widetilde H^{u_j}$. 
For each contour $\mathcal S^j=R_j+i\mathbb R$, we truncate the imaginary axis and approximate $\mathcal S^j$ by
$R_j+i[-L_j,L_j]$.
We then introduce the uniform grid
\[
u_{j,\ell}=R_j+i y_\ell, 
\qquad 
y_\ell=-L_j+\ell \Delta y_j,
\qquad 
\Delta y_j=\frac{2L_j}{N_j},
\quad \ell=0,\ldots,N_j.
\]
For every grid point $u_{j,\ell}$ and every time $t_k$, we compute
\[
\widetilde H_{t_k}^{u_{j,\ell}}
=
\exp\left(
u_{j,\ell}X_{t_k}
+
\int_{t_k}^T g_{u_{j,\ell}}(T-s)\xi_{t_k}(s)\,\rmd s
\right).
\]

The function $g_{u_{j,\ell}}$ is obtained from the Riccati--Volterra equation~\eqref{g 1u}. 
Rather than computing $g_u$ directly, we use the identity $g_u(t)=Q(u,\psi_u(t))$, as in \cite[Theorem~5.5]{larsson2023chapter}, where $\psi_u$ solves
\[
\psi_u(t)
=
\int_0^t K(t-s)F(u,\psi_u(s))\,\mathrm ds,
\]
with
$F(u,x)=\frac12(u^2-u)+(\rho\sigma u-\lambda)x+\frac12\sigma^2x^2, Q(u,x)=F(u,x)+\lambda x$.
Following \cite{abi2019affine}, $\psi_u$ is approximated by the multifactor system
$\psi_u^n(t)=\sum_{i=1}^n c_i^n \psi_u^{n,i}(t)$,
where
\[
\partial_t \psi_u^{n,i}
=
-\gamma_i^n\psi_u^{n,i}
+
F(u,\psi_u^n),
\qquad
\psi_u^{n,i}(0)=0,
\qquad i=1,\dots,n.
\]
Alternative numerical methods for the Riccati--Volterra equation include fractional Adams schemes 
\cite{diethelm2002predictor,diethelm2004detailed}, Padé approximation methods in the spirit of \cite{gatheral2019rational}, and hybrid methods such as those in \cite{callegaro2021fast}.

The computation of $\widetilde H_{t_k}^{u_{j,\ell}}$ is then performed pathwise. 
For each point $u_{j,\ell}$, we first solve the corresponding Riccati system on the time grid. 
We then approximate the forward variance curve $\xi_{t_k}(s)$, $s\in[t_k,T]$, using the simulated paths. 
In the implementation, $\xi$ is stored as a tensor of size $(N+1)\times(N+1)\times M$, whose $(k,r,m)$-entry corresponds to $\xi_{t_k}(t_r)$ on the $m$-th trajectory for $r\geq k$, and is set to zero otherwise. 
The integral in the definition of $\widetilde H_{t_k}^{u_{j,\ell}}$ is then evaluated by numerical quadrature.

We finally summarize how the quantities in \eqref{eq:strategies-frac-case} are computed.
The strategy $\theta^0$ is obtained by evaluating the corresponding coefficients on the
time grid and dividing by the simulated stock price $S^n$. The computation of
$\theta^j$ requires both $S^n$ and $\widetilde H^{u_{j,\ell}}$; the integral over the
complex contour is then evaluated by standard quadrature.

The coefficient $A$ is computed by evaluating the Mittag--Leffler function on the time
grid and applying numerical quadrature in time. For the matrices $B$ and $C$, the
expectations
\[
    \mathbb E[\widetilde H_t^{u_i}V_t],
    \qquad
    \mathbb E[\widetilde H_t^{u_i}\widetilde H_t^{u_j}V_t],
\]
are estimated by Monte Carlo averages over the simulated trajectories. These estimates
are then inserted into the expressions for $B$ and $C$, while the remaining time and
contour integrals are approximated by standard quadrature rules.

\paragraph{Parameters and numerical choices}

The model parameters used in the numerical experiments are reported in Table~\ref{tab:model_parameters}. 
They are chosen consistently with \cite{abi2019affine}. We do not report the Hurst index $H$, since different roughness regimes will be compared.
We verify that, for these parameters and for the maturities considered below, the moment conditions required in Theorem~\ref{thm:strategies-rough-heston} are satisfied. 
In particular, using Remark~\ref{rem:explosion-time-heston} and the comparison with the classical Heston explosion time in \cite{gerhold2019moment}, the maturity $T$ is chosen so that the relevant Fourier contours lie inside the admissible domain.

\begin{table}[h!]
\centering
\begin{tabular}{c|ccccccc}
\hline
Parameter 
& $\lambda$ & $\rho$ & $\sigma$ & $V_0$ & $X_0$ & $\theta$ \\
\hline
Value 
& $0.3$ & $-0.7$ & $0.3$ & $0.02$ & $\ln(100)$ & $0.066$ \\
\hline
\end{tabular}
\caption{Rough Heston parameters.}
\label{tab:model_parameters}
\end{table}

Table~\ref{tab:model_numeric_parameters} reports the numerical parameters used in the simulation: the number $M$ of Monte Carlo trajectories, the number $N$ of time steps, the number $n$ of factors in the multifactor approximation, and the parameters used for the truncation and discretization of the Fourier contours.

The static hedging instruments are chosen to be European call and put options. 
For these claims, we use the Fourier representations and the associated complex measures derived in Remark~\ref{rem:call-put-fourier}. We use one common contour for all calls,
$\mathcal S^C=R_C+i\mathbb R$ with $R_C>1$, and one common contour for all puts,
$\mathcal S^P=R_P+i\mathbb R$ with $R_P<0$.
The strike affects only the
measure $\zeta^i$.

In the numerical implementation, the call contour is truncated to
$R_C+i[-L_C,L_C],$
and the interval $[-L_C,L_C]$ is discretized uniformly with $N_C$ points. 
The same procedure is applied to the put contour, with truncation level $L_P$ and discretization size $N_P$. 

\begin{table}[h!]
\centering
\begin{tabular}{c|ccccccccc}
\hline
Parameter 
& $M$ & $N$ & $n$ & $N_C$ & $L_C$ & $R_C$ & $N_P$ & $L_P$ & $R_P$ \\
\hline
Value 
& $10^4$ & $250$ & $500$ & $100$ & $150$ & $2$ & $100$ & $150$ & $-0.1250$ \\
\hline
\end{tabular}
\caption{Numerical parameters used in the experiments.}
\label{tab:model_numeric_parameters}
\end{table}

The integrals defining the multifactor coefficients in \eqref{coeff multifactor approximation} are computed using the Python routine \verb|integrate.quad|. 
As recalled in Example~\ref{fractional kernel}, in the fractional case the resolvent of the second kind is expressed in terms of the Mittag--Leffler function. 
We evaluate it using the package \verb|pymittagleffler 0.2.1|, based on the algorithm of \cite{garrappa2015numerical}. 
Monte Carlo expectations are computed using \verb|np.mean|, while the time and contour integrals appearing in $A$, $B$, and $C$ are approximated using the trapezoidal rule, implemented through \verb|np.trapz|. 
For the parameter values in Table~\ref{tab:model_numeric_parameters}, the simulation of $M=10^4$ paths takes less than one minute, while the computation of $A$, $B$, and $C$ takes approximately one hour.

\subsection{Optimal selection}

We numerically solve the optimal selection problem introduced in
Section~\ref{sec:optimal-selection}. The set of candidate contingent claims consists of
European put and call options with strikes spaced by $5$. More precisely, we consider
put options with strikes $K \in \{50,55,\ldots,95\}$,
and call options with strikes
$K \in \{100,105,\ldots,150\}$.
In total, we have $d=21$ candidate instruments.

For each prescribed basket size $m$, we perform a search over all subsets
of cardinality $m$. Namely, among the $\binom{21}{m}$ admissible combinations, we
select the subset which yields the smallest residual hedging error. The analysis is
carried out for three values of the Hurst parameter, $H=0.1, H=0.3, H=0.5.$
The last case corresponds to the classical Heston model. In this Markovian setting,
closed-form solutions to the Riccati equations are available, together with explicit formulae for the expectation terms appearing in the matrices $B$ and $C$
see, for instance, \cite{ditella2020semistatic}. We exploit this structure in the
computation of the optimal strategy for $H=0.5$, employing the same computational methods as in  \cite{ditella2020semistatic}. We note that computaitonal time for optimal claim selection is around 10 minutes for every set of parameters, taking into account the code has been parallelized.

\paragraph{Neuberger's replication formula}
Before illustrating the numerical results we recall the model-free replication formula of
Neuberger~\cite{neuberger1994log}. In a continuous time, zero-interest-rate setting,
Neuberger formula reads
\begin{equation}\label{eq:neu-formula}
[X,X]_T
=
2\int_0^T \left(\frac{1}{S_t} - \frac{1}{S_0} \right)dS_t
+
2\int_0^{S_0}\frac{(K-S_T)^+}{K^2}\,dK
+
2\int_{S_0}^{\infty}\frac{(S_T-K)^+}{K^2}\,dK.
\end{equation}
The equality indicates that the semi-static replication of a variance swap
is obtained from a continuum of out-of-the-money options with strikes $K$ ranging from $0$ to infinity and  positive weights
proportional to $K^{-2}$. In contrast, we use $21$ strikes between $50$ and $150$, selecting weights by minimizing the quadratic hedging error. Even though we consider the discretized problem, since the rough Heston model has continuous paths,
we expect the optimal portfolios to reproduce a behavior consistent with the Neuberger formula.

\paragraph{Claim selection}
We perform two types of experiments. In the first one, both long and short positions are
allowed. In the second one, mimicking the Neuberger formula, we impose the additional constraint that only long positions
are admissible. The corresponding results are reported in
Figure~\ref{fig:optimal_sel_unconstrained} and Figure~\ref{fig:optimal_sel_longonly},
respectively.

\begin{figure}[!h]
    \centering
    \includegraphics[width=0.32\textwidth]{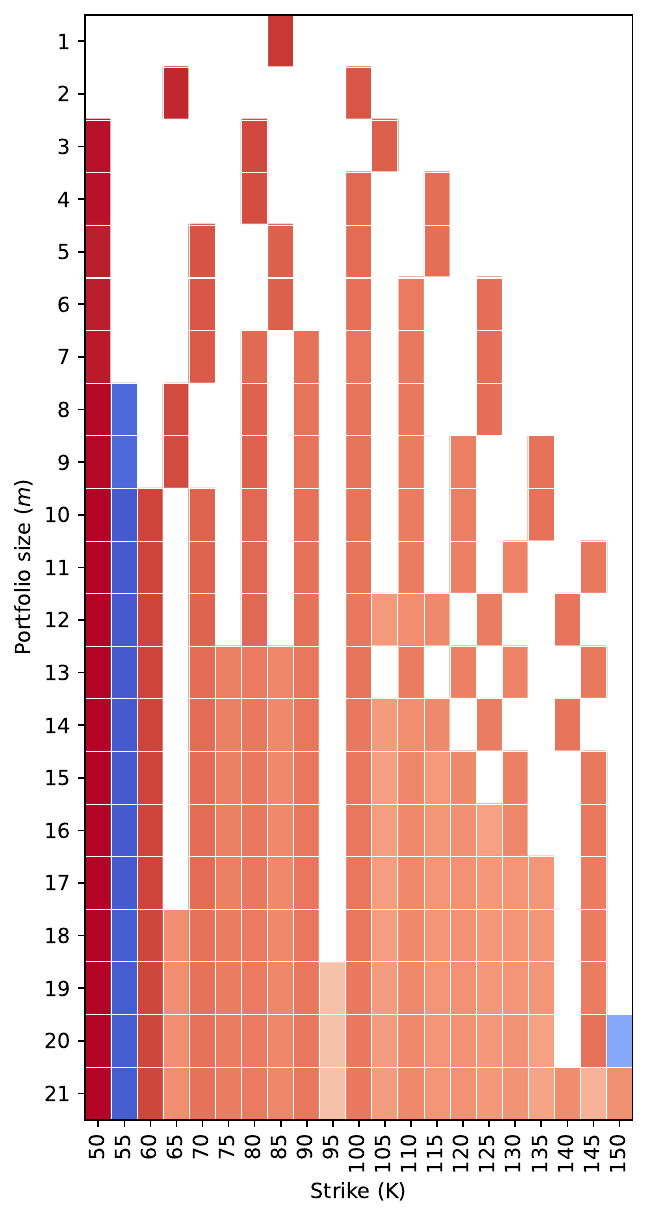}
    \includegraphics[width=0.32\textwidth]{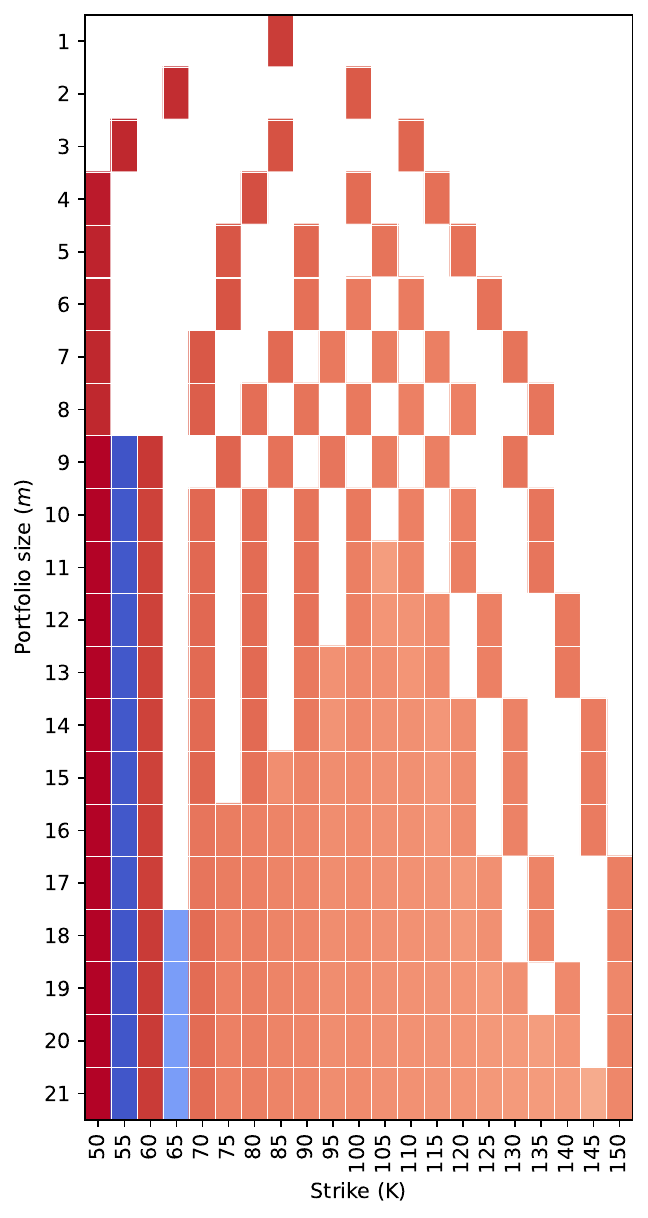}
    \includegraphics[width=0.32\textwidth]{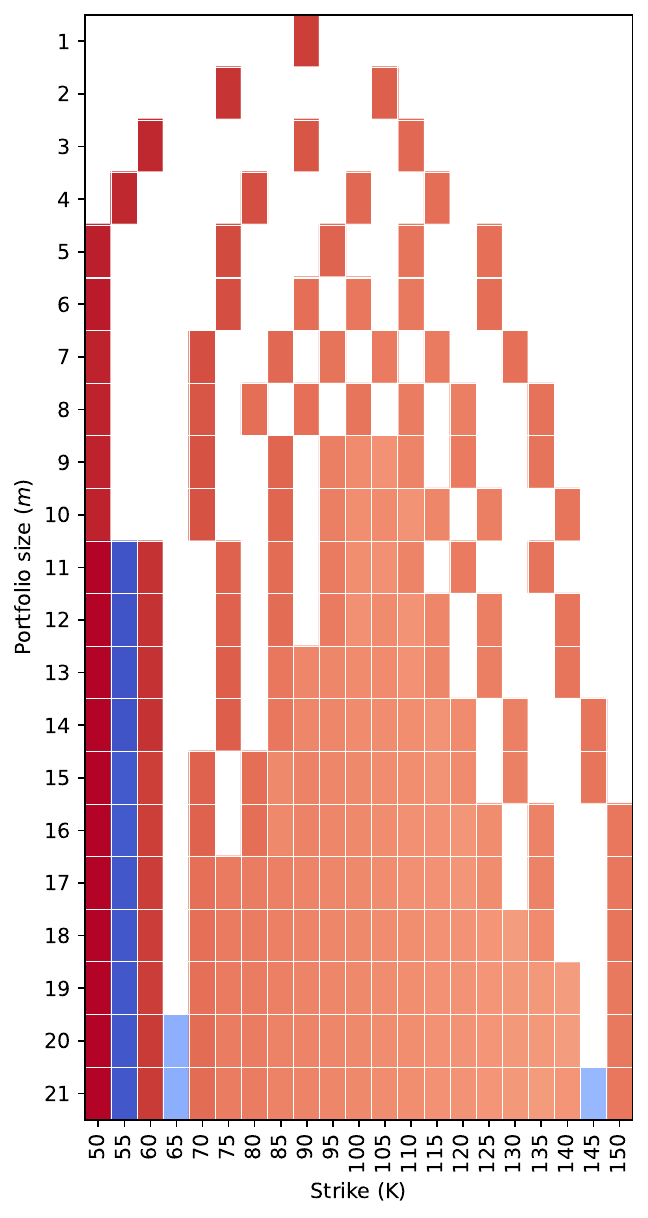}
    \caption{Composition of the optimal static portfolio as a function of the basket size
    $m$, for $H=0.1$, $H=0.3$, and $H=0.5$. Long and short positions are allowed.}
    \label{fig:optimal_sel_unconstrained}
\end{figure}

\begin{figure}[!h]
    \centering
    \includegraphics[width=0.32\textwidth]{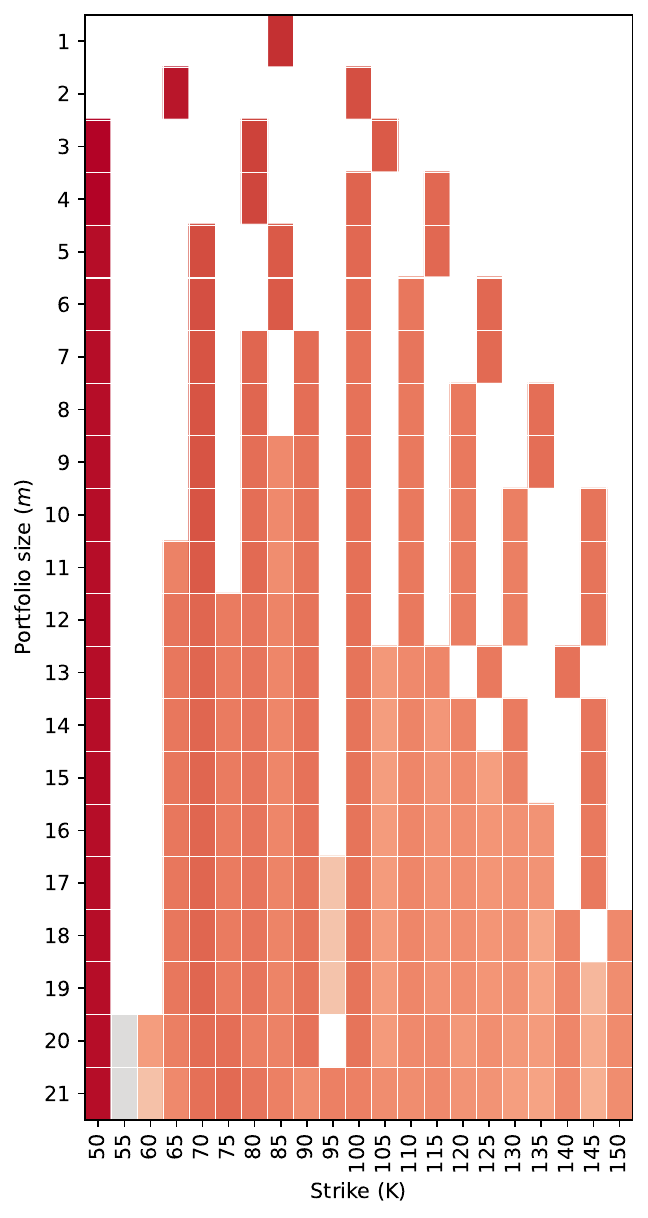}
    \includegraphics[width=0.32\textwidth]{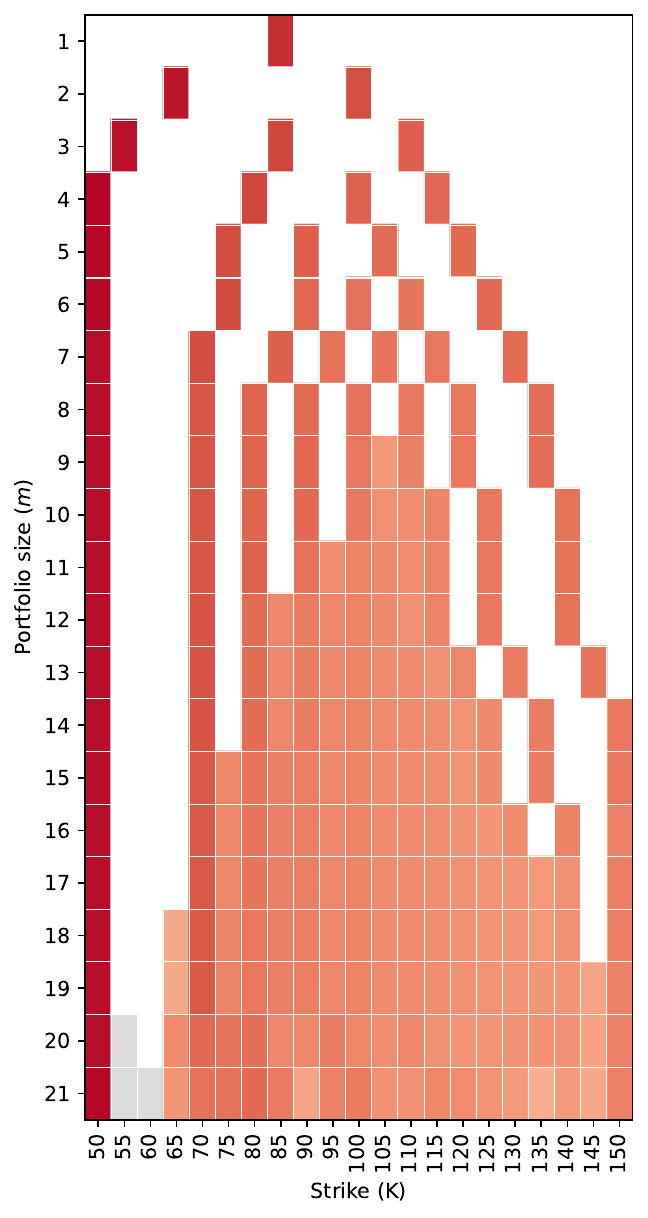}
    \includegraphics[width=0.32\textwidth]{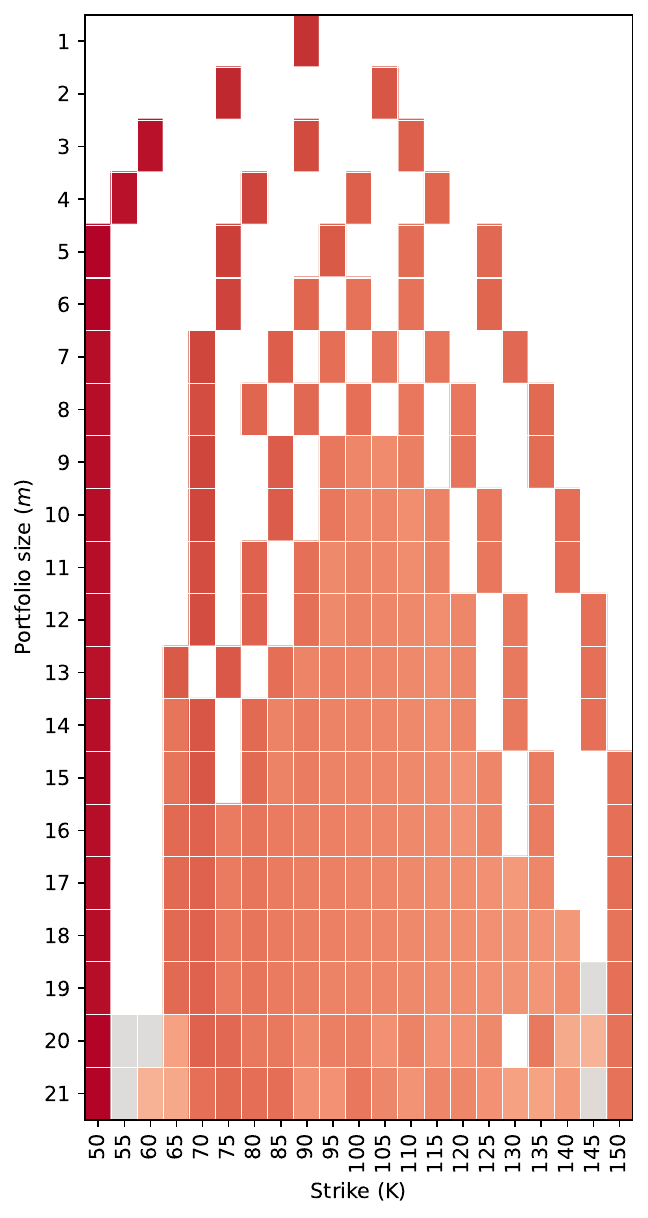}
    \caption{Composition of the optimal static portfolio as a function of the basket size
    $m$, for $H=0.1$, $H=0.3$, and $H=0.5$. Only long positions are allowed.}
    \label{fig:optimal_sel_longonly}
\end{figure}

The unconstrained solutions are mostly characterized by long positions. Short positions
appear only at a small number of strikes, mostly out-of-the-money. These short positions are isolated and we numerically verify that they do not strongly reduce the
overall resulting hedging error. As observed in
\cite{ditella2020semistatic}, this suggests that their occurrence is mainly due to
numerical artifacts. For this reason, it is natural to also consider the long-only
specification, which is financially interpretable in the sense of the Neuberger replication in \eqref{eq:neu-formula}.

Across all three values of the Hurst parameter, the optimal portfolios assign larger
positions to out-of-the-money puts than to out-of-the-money calls. Moreover, the
selection pattern is qualitatively similar across the three roughness regimes. For small
basket sizes, the selected options are typically concentrated around the at-the-money
region. As $m$ increases, the selected strikes progressively move towards both tails
of the strike grid. This expansion is more pronounced on the put side: put options are
selected earlier and usually receive larger positions.

A further difference concerns the speed at which the algorithm includes the lowest
available put strikes. In the roughest case, $H=0.1$, the lowest put strike enters the
optimal portfolio already for $m=3$. For $H=0.3$, this occurs for $m=4$, while in
the classical Heston case $H=0.5$, it occurs only for $m=5$.

\paragraph{Residual hedging error}
Figure~\ref{fig:error} reports, on a log scale, the optimal residual
hedging error as a function of the basket size $m$ for the long-only setting.
For notational convenience, we denote this error by $\varepsilon^2(m)$
throughout this section.
As expected, the error decreases as $m$ increases,
since enlarging the static portfolio expands the admissible set of hedging strategies.
The decrease is typically more pronounced for small values of $m$, where the addition
of one option has a substantial marginal effect. For larger basket sizes, the curve
flattens, suggesting that most of the relevant hedging contribution is already captured
by a relatively small number of well-selected instruments. To complement
Figure~\ref{fig:error}, Table~\ref{tab:hedging-errors} reports the
optimal residual hedging error $\varepsilon^2(m)$ for selected basket sizes,
namely $m=0,1,2,3$ and $m=12$, and for the three values of the Hurst parameter
considered in the numerical experiment. The case $m=0$ corresponds to the purely
dynamic strategy, $\varepsilon^2(0)=A$.

\begin{table}[!ht]
\centering
\begin{tabular}{c|ccccc}
\toprule
$H$ 
& $m=0$
& $m=1$ 
& $m=2$ 
& $m=3$ 
& $m=12$ \\
\midrule
$0.1$ 
& $4.8549\times 10^{-4}$
& $7.0312\times 10^{-5}$
& $2.8705\times 10^{-5}$
& $1.8352\times 10^{-5}$
& $1.3933\times 10^{-5}$ \\
$0.3$
& $3.8056\times 10^{-4}$
& $3.9411\times 10^{-5}$
& $1.2343\times 10^{-5}$
& $5.9245\times 10^{-6}$
& $2.5548\times 10^{-6}$ \\
$0.5$
& $2.8754\times 10^{-4}$
& $2.3140\times 10^{-5}$
& $5.9273\times 10^{-6}$
& $2.0377\times 10^{-6}$
& $9.5131\times 10^{-8}$ \\
\bottomrule
\end{tabular}
\caption{Optimal residual hedging error $\varepsilon^2(m)$ for selected basket sizes
and for the three Hurst regimes. The case $m=0$ corresponds to the purely dynamic
strategy, $\varepsilon^2(0)=A$.}
\label{tab:hedging-errors}
\end{table}
We observe that lower values of $H$ are typically associated
with larger residual hedging errors. This is consistent with the intuition that, in
rougher regimes, the variance process exhibits more irregular short-time behaviour,
making the residual volatility risk harder to hedge.
Moreover, the table suggests that the introduction of a small number of static claims is enough to
produce a substantial reduction of the hedging error.
To quantify this reduction, we define the hedging effectiveness
\begin{equation}\label{eq:delta_m}
\Delta(m) = \left(1-\frac{\varepsilon^2(m)}{A}\right)\times 100,
\end{equation}
which measures the percentage decrease in the hedging error obtained by including
the best $m$ static instruments, relative to the purely dynamic baseline
$A$. Figure~\ref{fig:error_rel} displays $\Delta(m)$ as a
function of $m$, while the corresponding values for selected basket sizes are
reported in Table~\ref{tab:hedging-effectiveness}.

\begin{table}[!ht]
\centering
\begin{tabular}{c|cccc}
\toprule
$H$ 
& $m=1$ 
& $m=2$ 
& $m=3$ 
& $m=12$ \\
\midrule
$0.1$ 
& $85.52$
& $94.09$
& $96.22$
& $97.13$ \\
$0.3$
& $89.64$ & $96.77$ & $98.44$
& $99.32$ \\
$0.5$
& $91.95$ & $97.94$&$99.29$
& $99.96$ \\
\bottomrule
\end{tabular}
\caption{Hedging effectiveness $\Delta(m)$, in percentage, for selected basket
sizes and for the three Hurst regimes.}
\label{tab:hedging-effectiveness}
\end{table}

In all three Hurst regimes, we conclude that "3 options are roughly enough", i.e. a hedging effectiveness above $95\%$ is already
achieved with $m=3$ static instruments, showing that the addition of only a limited number of well-selected options leads to a substantial improvement of the hedging
error. Starting from $m=5$, the improvement flattens, stabilizing around $97\%$
for $H=0.1$ and $99\%$ for $H=0.3$ and $H=0.5$.

To further assess how close the error is to its limiting value, we consider the
normalized hedging effectiveness
\begin{equation}\label{eq:eta_m}
\eta(m) = \left(\frac{A-\varepsilon^2(m)}
               {A-\varepsilon^2(21)}\right)\times 100,
\end{equation}
which measures the percentage of the total achievable error reduction that is
captured with $m$ static instruments. Here $\varepsilon^2(21)$ is used as a
proxy for the asymptotic error, since the curve in Figure~\ref{fig:error}
stabilizes for large $m$. By construction, $\eta(m)$ equals $0$ for $m=0$ and increases to $100$ as $m$ approaches the limiting basket size; large values of
$\eta(m)$ indicates that a large fraction of the total achievable error reduction is already captured by a basket of $m$ static instruments. Figure~\ref{fig:error_effect} displays $\eta(m)$ as a function of
$m$, while the corresponding values for selected basket sizes are reported in
Table~\ref{tab:hedging-efficiency}.

\begin{table}[!ht]
\centering
\begin{tabular}{c|ccc}
\toprule
$H$
& $m=1$
& $m=2$
& $m=3$ \\
\midrule
$0.1$ & 88.04 &96.87 &  99.06 \\
$0.3$ & 90.25 &  97.411 & 99.11 \\
$0.5$ & 91.98 & 97.97 &  99.32 \\
\bottomrule
\end{tabular}
\caption{Normalized hedging effectiveness $\eta(m)$ for selected basket sizes
and for the three Hurst regimes.}
\label{tab:hedging-efficiency}
\end{table}

The values of $\eta(m)$ confirm the flattening of the improvement observed
in Figure~\ref{fig:error}: in all three Hurst regimes, with $m=2$ static
instruments more than $95\%$ of the total achievable error reduction is
already captured, and with $m=3$ more than $99\%$. Beyond this point,
adding further instruments yields only marginal gains.

\paragraph{Optimal strategies for $m=3$}
As a further experiment, we analyze in more detail the optimal static and dynamic
strategies for a portfolio of size $m=3$. We again compare the three roughness
regimes $H=0.1$, $H=0.3$, and $H=0.5$.

Figure~\ref{fig:optimal_static_m3} displays the optimal static positions in a log-log plot,
with the strike on the horizontal axis and the corresponding portfolio weight on the
vertical axis. Across all three values of $H$, the optimal weights exhibit a slope that
is close to the theoretical slope implied by $\nu(K)$ proportional to $K^{-2}$ as in \eqref{eq:neu-formula}.
Indeed, taking logarithms gives
\[
    \log \nu(K) = \mathrm{const} - 2 \log K,
\]
so that the benchmark slope in the log-log plot is equal to $-2$. The numerical
results are consistent with this prediction, confirming that even when only a small
number of options is selected, the optimal static component retains the main structural
feature of Neuberger's continuous-strike formula.

Finally, we compare the optimal dynamic strategies
\[
    \vartheta^*
    =
    \vartheta^0-\sum_{j=1}^3 \nu_j^* \vartheta^j.
\]
The price paths are simulated using the same Brownian motions in
\eqref{path x} and \eqref{eq:sim-V^n}, so that the resulting strategies are directly
comparable across different values of $H$. In
Figure~\ref{fig:dyn_strat} we plot the optimal dynamic strategy for a single in the three Hurst regimes. In Figures~\ref{fig:dyn_strat_diff_01} and ~\ref{fig:dyn_strat_diff_03} we plot the relative difference between the optimal rough Heston dynamic strategy and the classical one 

\[\frac{  \vartheta^*-  \vartheta^*_{\text{Heston}}}{\vartheta^*_{\text{Heston}}},\]
for three different realized paths of $(S,V)$.

The dynamic component is more sensitive to the roughness parameter than the static one.
While it looks like the optimal static portfolios remain broadly similar across the three regimes, the
dynamic strategy becomes larger in absolute value as $H$ decreases. In particular, for $H=0.1$ the relative difference ranges between roughly $7\%$ and
$110\%$, while for $H=0.3$ it reaches values between 
$4\%$ and $50\%$: in the roughest regime the position in the underlying asset can be
more than twice as large as in the classical Heston case. These results suggest that the static part of the semi-static hedge appears largely insensitive to the roughness
parameter, while the dynamic part is more sensitive to the additional pathwise
fluctuations induced by rough volatility.

\begin{figure}[htbp]
    \centering

    \begin{minipage}{\linewidth}
        \centering
        \includegraphics[width=0.8\linewidth]{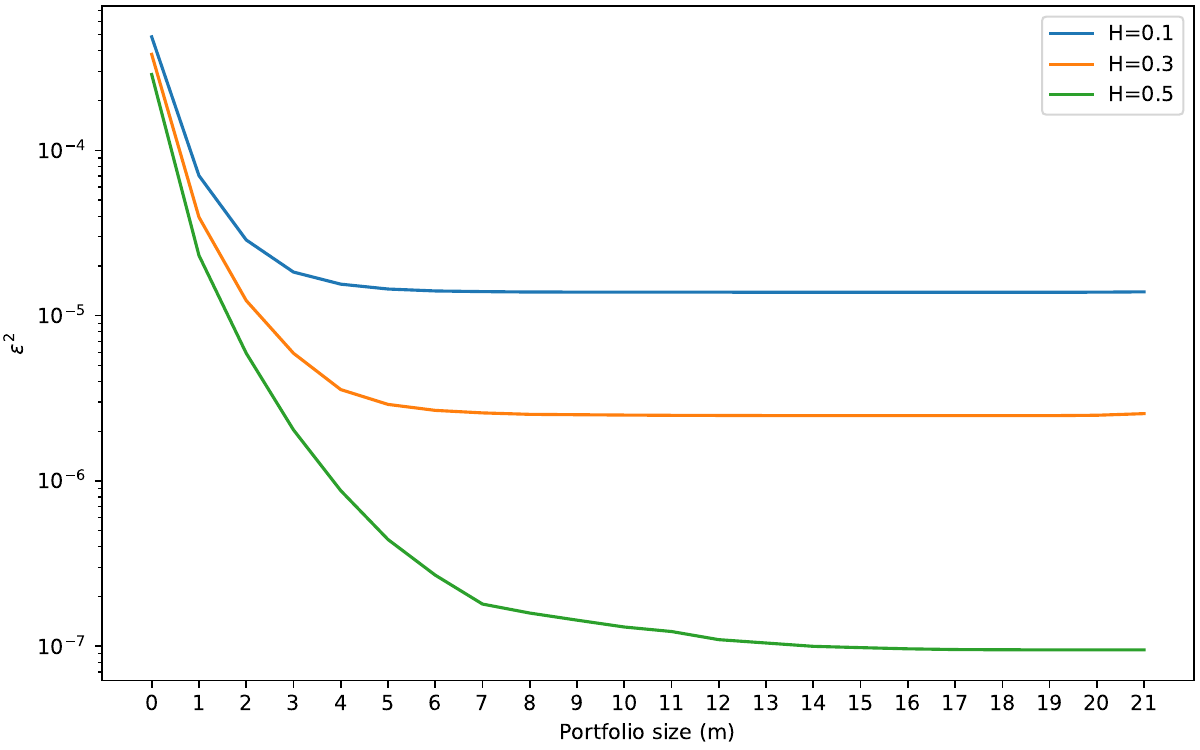}
        \caption{Optimal residual hedging error $\varepsilon^2(m)$ in log-plot for selected basket sizes
and for the three Hurst regimes.}
        \label{fig:error}
    \end{minipage}

    \vspace{0.5cm}

    \begin{minipage}{\linewidth}
        \centering
        \includegraphics[width=0.8\linewidth]{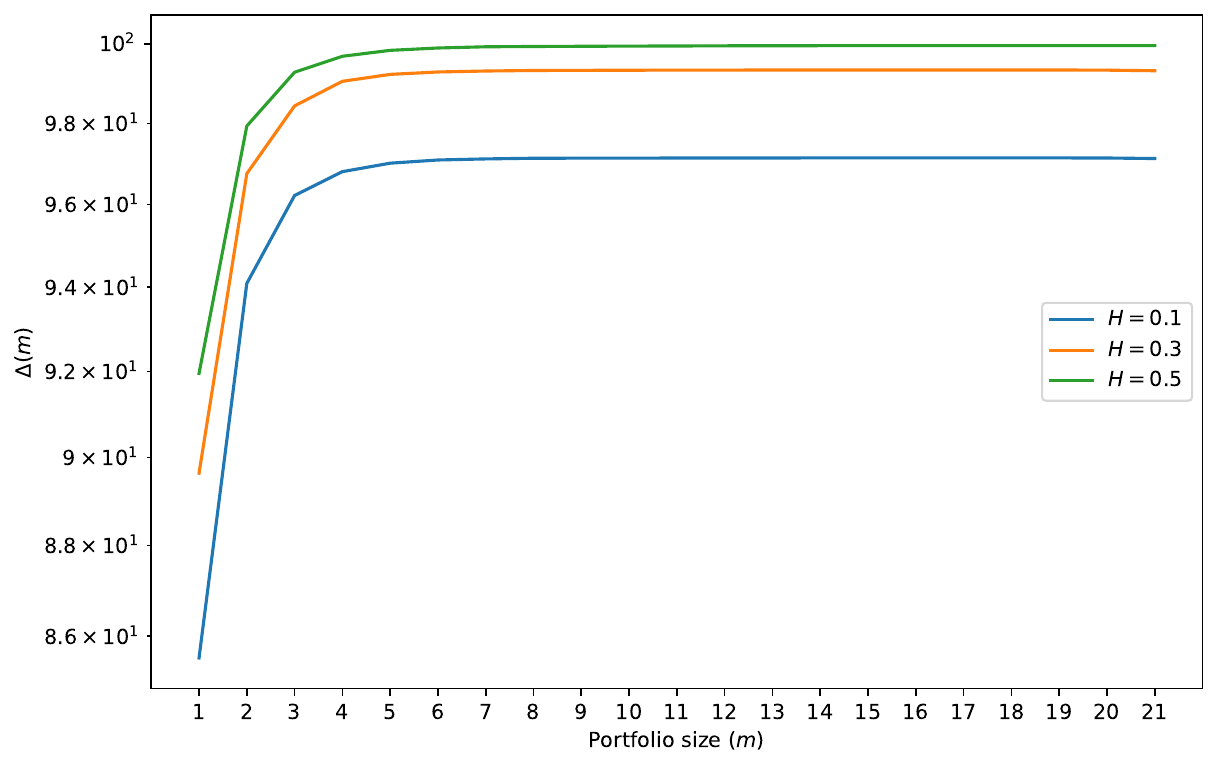}
        \caption{Hedging effectiveness $\Delta(m)$, equation~\eqref{eq:delta_m}, for the three Hurst regimes.}
        \label{fig:error_rel}
    \end{minipage}

\end{figure}

\begin{figure}

    \begin{minipage}{\linewidth}
        \centering
        \includegraphics[width=0.8\linewidth]{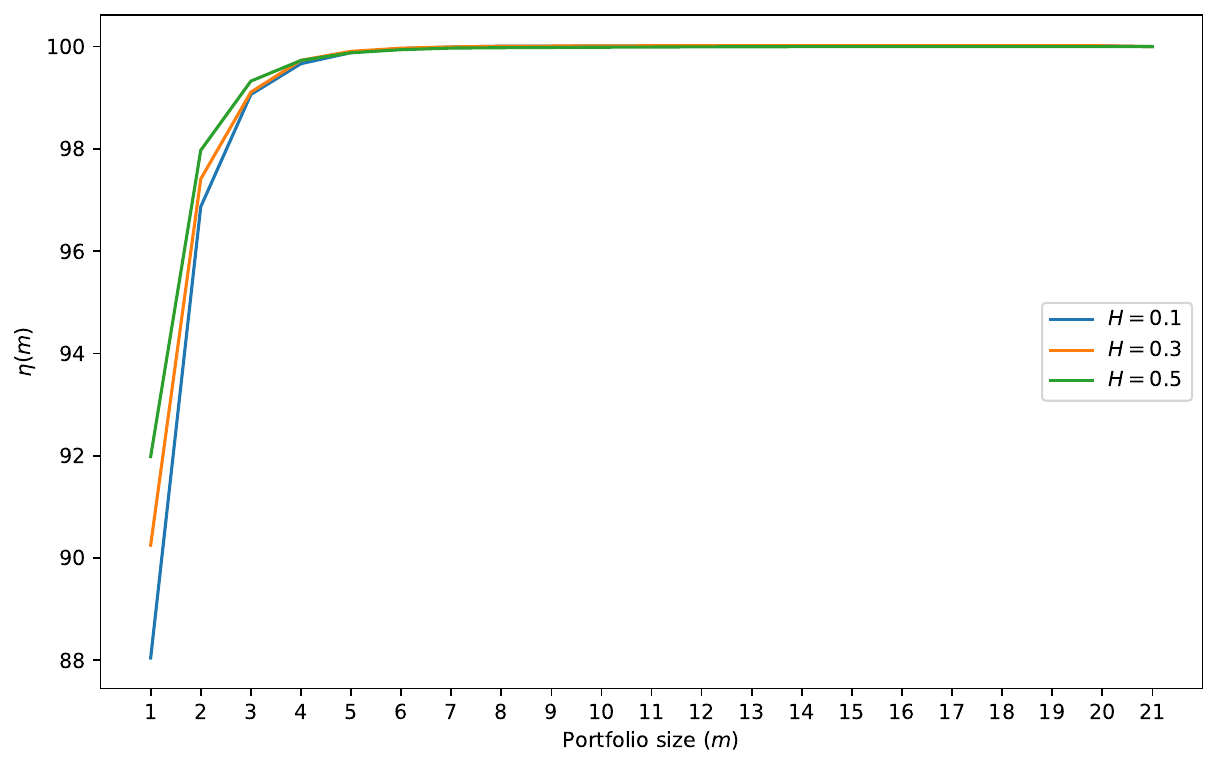}
        \caption{Normalized hedging effectiveness $\eta(m)$, equation~\eqref{eq:eta_m}, for the three Hurst regimes.}
    \end{minipage}
\label{fig:error_effect}
\end{figure}

\begin{figure}
    \centering

    \includegraphics[width=0.9\linewidth]{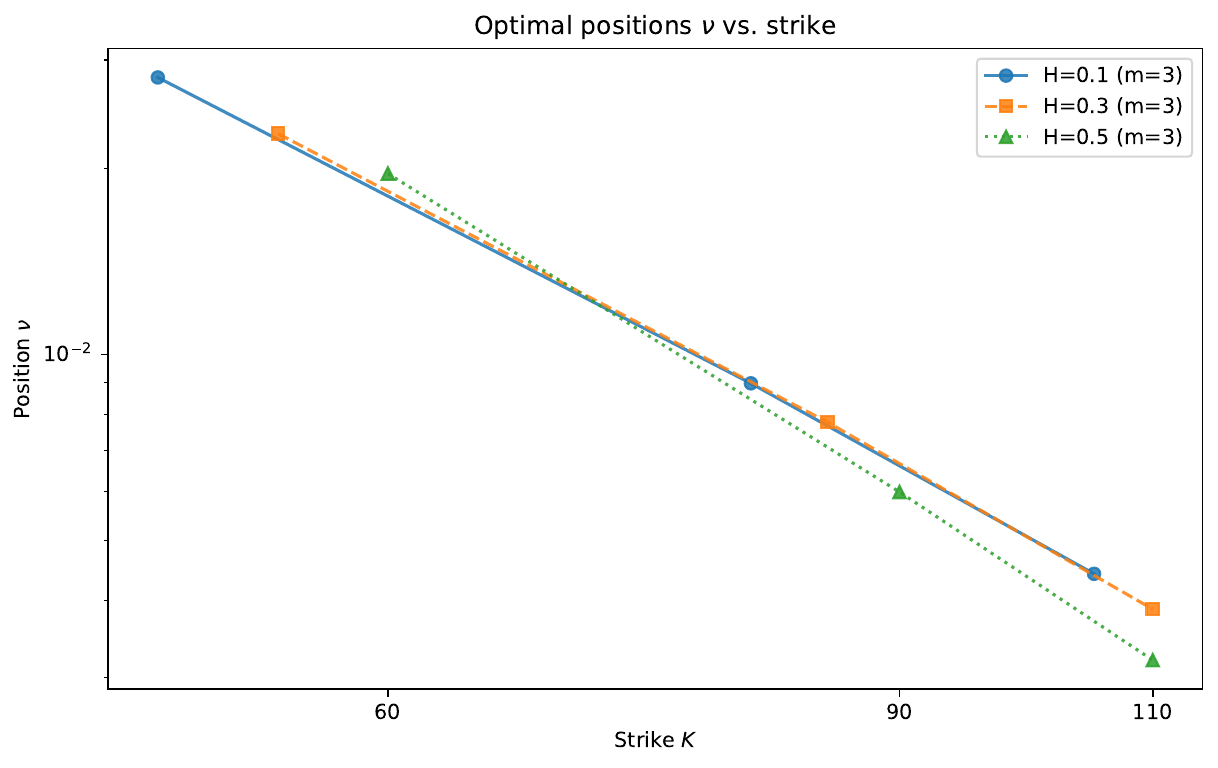}
    \caption{Portfolio weights $\nu(K)$ in log-log plot in the optimal hedging portfolios of effective size $m=3$.}
    \label{fig:optimal_static_m3}

\end{figure}

\begin{figure}
    \centering
    \includegraphics[width=0.8\linewidth]{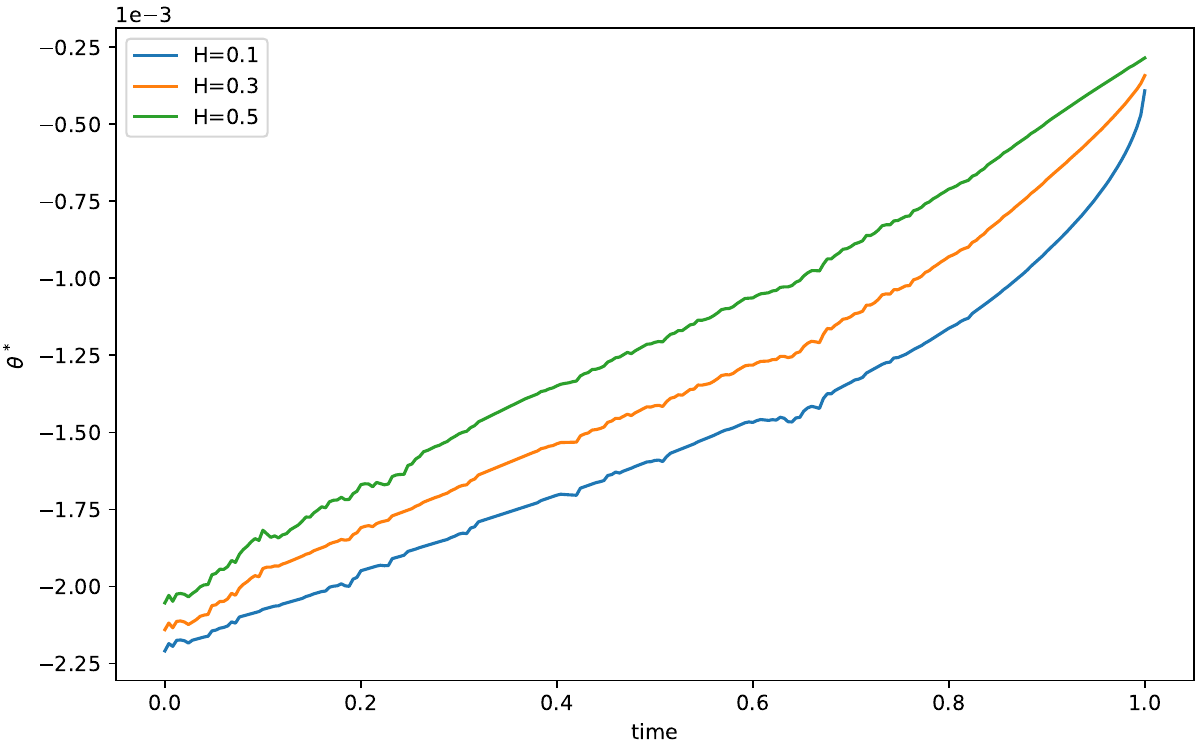}
    \caption{Optimal dynamic strategies $\theta^*$ in the optimal hedging portfolio of effective size $m=3$ for a single realized path of $(S,V)$.}
    \label{fig:dyn_strat}
\end{figure}

\begin{figure}[htbp]
    \centering

    \begin{minipage}{\linewidth}
        \centering
        \includegraphics[width=0.8\linewidth]{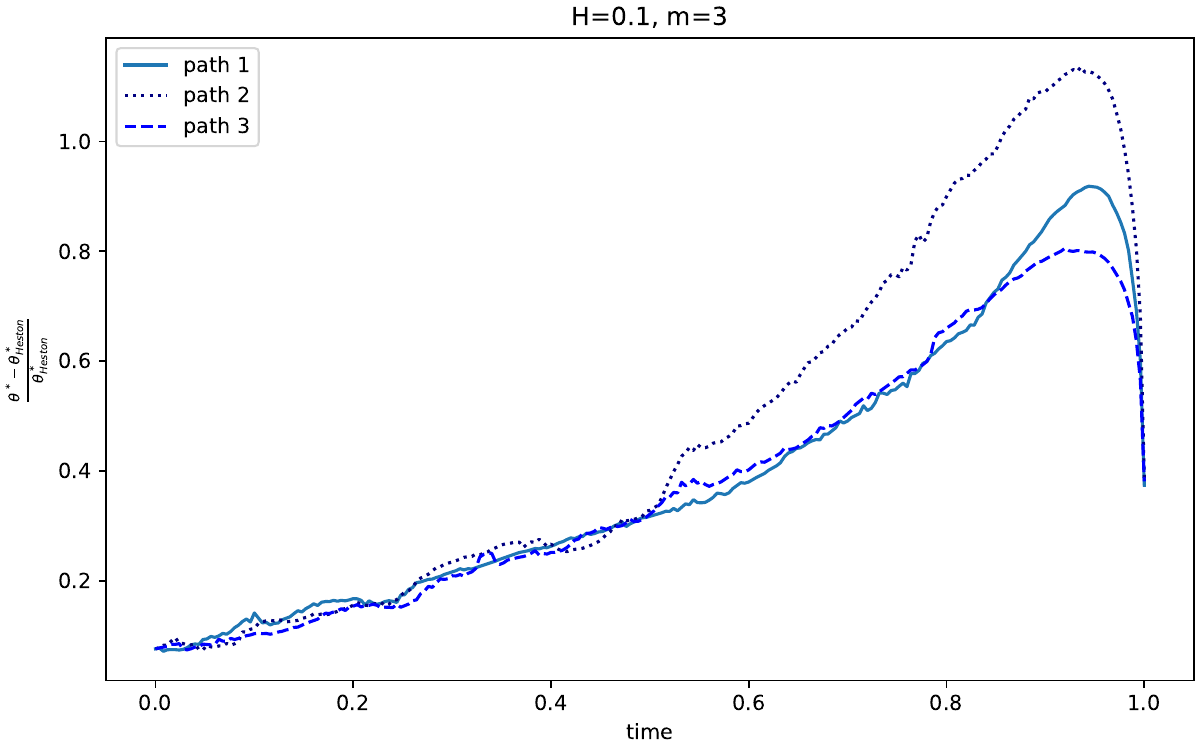}
        \caption{Relative dynamic hedge difference in the optimal hedging portfolio of effective size $m=3$ for $H=0.1$ and three realized paths of $(S,V)$.}
        \label{fig:dyn_strat_diff_01}
    \end{minipage}

    \vspace{0.5cm}

    \begin{minipage}{\linewidth}
        \centering
        \includegraphics[width=0.8\linewidth]{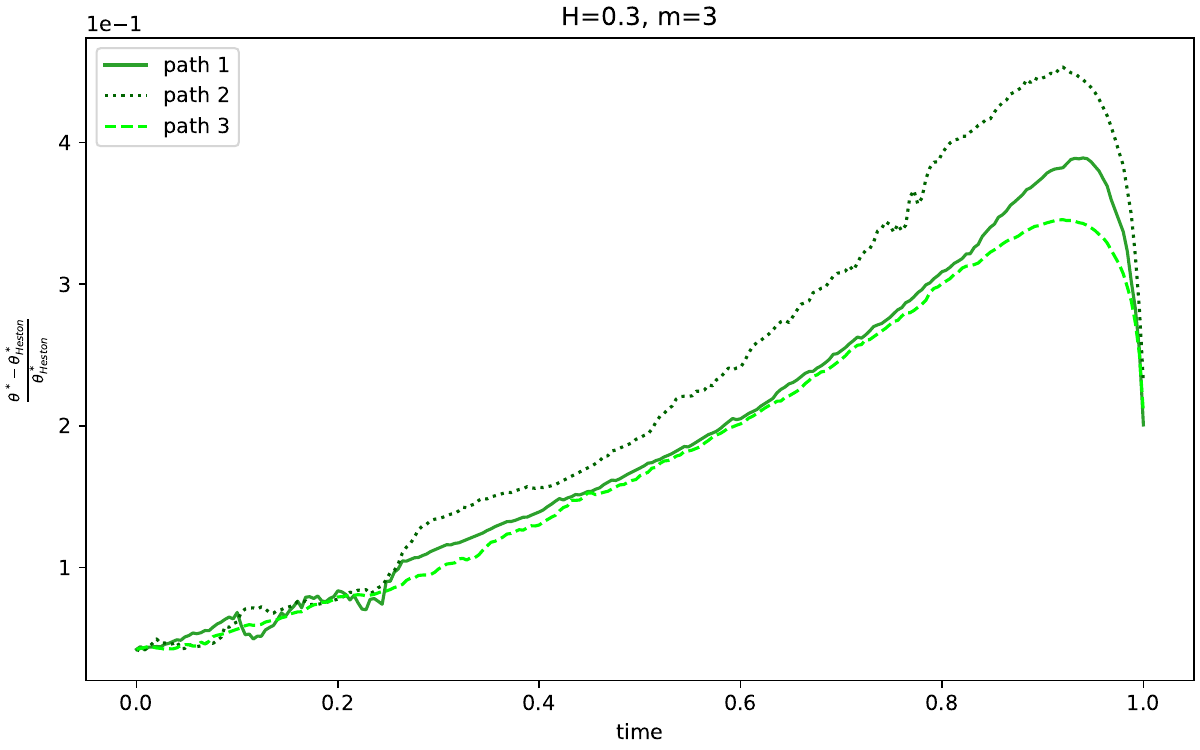}
        \caption{Relative dynamic hedge difference in the optimal hedging portfolio of effective size $m=3$ for  $H=0.3$ and three realized paths of $(S,V)$.}
        \label{fig:dyn_strat_diff_03}
    \end{minipage}

\end{figure}

\section{Conclusions}\label{sec:conclusions}

We have derived semi-static variance-optimal hedging strategies for a variance
swap in the rough Heston model. The hedge combines dynamic trading in the
underlying asset with static positions in a finite basket of European options.
By combining the affine transform formula of the rough Heston model with the
Fourier--GKW methodology of \cite{ditella2019semistatic}, we obtained
semi-explicit formulas for the optimal hedging strategies.

The resulting formulas are expressed in terms of the forward variance curve,
the Riccati--Volterra solution and the conditional Laplace transform. The strategies are numerically computed by means of a multifactor
approximation of the fractional kernel, combined with Monte-Carlo simulations and standard quadrature methods.

Three findings emerge from the numerical analysis. First, a basket of three
optimally selected options reduces the hedging error by more than $95\%$
compared to the purely dynamic hedge in all roughness regimes considered, and
enlarging the basket further yields only marginal improvements. The selected
instruments are out-of-the-money options with weights decaying approximately
as $K^{-2}$, consistently with Neuberger variance swap replication.

Second, roughness affects the size of the residual error rather than the
composition of the basket: the optimal static portfolio is stable with respect to the Hurst index, while the residual error becomes larger in
rougher regimes.

Third, the dynamic component of the hedge becomes larger in magnitude as the
Hurst index decreases: rougher regimes require heavier positions in
the underlying asset in order to offset the faster fluctuations of the
variance process.

Possible extensions of this work consist in considering the same problem under
different rough models, for instance in the setting of rough affine models
with jumps as in \cite{bondi2024rough}, or in extending the techniques to
models whose characteristic function is not available explicitly, such as rough
Bergomi \cite{bayer2016pricing}, or rough SABR \cite{fukasawa2021rough}.

\section*{AI Disclosure}

Generative AI was used for language improvement and coding assistance.

\clearpage
\bibliographystyle{amsplain}
\bibliography{biblio}

\providecommand{\bysame}{\leavevmode\hbox to3em{\hrulefill}\thinspace}
\providecommand{\MR}{\relax\ifhmode\unskip\space\fi MR }
\providecommand{\MRhref}[2]{%
  \href{http://www.ams.org/mathscinet-getitem?mr=#1}{#2}
}
\providecommand{\href}[2]{#2}
\begin{thebibliography}{10}

\bibitem{abi2019lifting}
Eduardo Abi~Jaber, \emph{Lifting the heston model}, Quantitative finance \textbf{19} (2019), no.~12, 1995--2013.

\bibitem{abi2019multifactor}
Eduardo Abi~Jaber and Omar El~Euch, \emph{Multifactor approximation of rough volatility models}, SIAM journal on financial mathematics \textbf{10} (2019), no.~2, 309--349.

\bibitem{abi2019affine}
Eduardo Abi~Jaber, Martin Larsson, and Sergio Pulido, \emph{Affine volterra processes}, The Annals of Applied Probability \textbf{29} (2019), no.~5, 3155–3200.

\bibitem{acciaio2017semi}
Beatrice Acciaio and Martin Larsson, \emph{Semi-static completeness and robust pricing by informed investors}, The Annals of Applied Probability \textbf{27} (2017), no.~4, 2270--2300.

\bibitem{acciaio2017space}
Beatrice Acciaio, Martin Larsson, and Walter Schachermayer, \emph{The space of outcomes of semi-static trading strategies need not be closed}, Finance and Stochastics \textbf{21} (2017), no.~3, 741--751.

\bibitem{andersen2007moment}
Leif~B.G. Andersen and Vladimir~V. Piterbarg, \emph{{Moment Explosions in Stochastic Volatility Models}}, Finance and Stochastics \textbf{11} (2007), 29--50.

\bibitem{bayer2016pricing}
Christian Bayer, Peter Friz, and Jim Gatheral, \emph{Pricing under rough volatility}, Quantitative Finance \textbf{16} (2016), no.~6, 887--904, DOI: 10.1080/14697688.2015.1099717.

\bibitem{beiglbock2013model}
Mathias Beiglb{\"o}ck, Pierre Henry-Labordère, and Friedrich Penkner, \emph{{Model-independent bounds for option prices: A mass transport approach.}}, Finance and Stochastics \textbf{17} (2013), no.~3, 477--501.

\bibitem{bondi2024rough}
Alessandro Bondi, Sergio Pulido, and Simone Scotti, \emph{The rough hawkes heston stochastic volatility model}, Mathematical Finance \textbf{34} (2024), no.~4, 1197--1241.

\bibitem{callegaro2026semistatic}
Giorgia Callegaro, Paolo Di~Tella, Beatrice Ongarato, and Carlo Sgarra, \emph{Semistatic variance-optimal hedging with self-exciting jumps}, Mathematics of Operations Research \textbf{51} (2026), no.~3, 2495--2524.

\bibitem{callegaro2021fast}
Giorgia Callegaro, Martino Grasselli, and Gilles Pages, \emph{Fast hybrid schemes for fractional riccati equations (rough is not so tough)}, Mathematics of Operations Research \textbf{46} (2021), no.~1, 221--254.

\bibitem{chatziandreou2026pricing}
Konstantinos Chatziandreou and Sven Karbach, \emph{Pricing and semi-static hedging of green pay-as-produced power purchase agreements}, arXiv preprint arXiv:2607.27814 (2026).

\bibitem{chatziandreou2026semi}
\bysame, \emph{Semi-static variance-optimal hedging of covariance risk in multi-asset derivatives}, arXiv preprint arXiv:2603.25320 (2026).

\bibitem{cuchiero2020generalized}
Christa Cuchiero and Josef Teichmann, \emph{Generalized feller processes and markovian lifts of stochastic volterra processes: the affine case}, Journal of evolution equations \textbf{20} (2020), no.~4, 1301--1348.

\bibitem{ditella2019semistatic}
Paolo Di~Tella, Martin Haubold, and Martin Keller-Ressel, \emph{Semi-static variance-optimal hedging in stochastic volatility models with {F}ourier representation}, Journal of Applied Probability \textbf{56} (2019), no.~3, 787--809.

\bibitem{ditella2020semistatic}
\bysame, \emph{Semistatic and sparse variance-optimal hedging}, Mathematical Finance \textbf{30} (2020), no.~2, 403--425.

\bibitem{diethelm2002predictor}
Kai Diethelm, Neville~J Ford, and Alan~D Freed, \emph{A predictor-corrector approach for the numerical solution of fractional differential equations}, Nonlinear Dynamics \textbf{29} (2002), no.~1, 3--22.

\bibitem{diethelm2004detailed}
\bysame, \emph{Detailed error analysis for a fractional adams method}, Numerical algorithms \textbf{36} (2004), no.~1, 31--52.

\bibitem{duffie2003affine}
Darrell Duffie, Damir Filipovi{\'c}, and Walter Schachermayer, \emph{Affine processes and applications in finance}, The Annals of Applied Probability \textbf{13} (2003), no.~3, 984--1053, DOI: 10.1214/aoap/1060202833.

\bibitem{el2019characteristic}
Omar El~Euch and Mathieu Rosenbaum, \emph{The characteristic function of rough heston models}, Mathematical Finance \textbf{29} (2019), no.~1, 3--38.

\bibitem{fukasawa2021rough}
Masaaki Fukasawa and Jim Gatheral, \emph{A rough sabr formula}, Frontiers of Mathematical Finance \textbf{1} (2022), no.~1, 81--97, DOI: 10.3934/fmf.2021003.

\bibitem{follmer1986hedging}
Hans Föllmer and Dieter Sondermann, \emph{Hedging of non-redundant contingent claims}, Contributions to mathematical economics (W.~Hildenbrand and A.~Mas-Colell, eds.), Elsevier Science Publishing, North Holland, The Netherlands, 1986, pp.~205--223.

\bibitem{garrappa2015numerical}
Roberto Garrappa, \emph{Numerical evaluation of two and three parameter mittag-leffler functions}, SIAM Journal on Numerical Analysis \textbf{53} (2015), no.~3, 1350--1369.

\bibitem{gatheral2022volatility}
Jim Gatheral, Thibault Jaisson, and Mathieu Rosenbaum, \emph{Volatility is rough}, Quantitative Finance \textbf{18} (2018), no.~6, 933--949, DOI: 10.1080/14697688.2017.1393551.

\bibitem{gatheral2019rational}
Jim Gatheral and Rado{\v{s}} Radoi{\v{c}}i{\'c}, \emph{Rational approximation of the rough heston solution}, International Journal of Theoretical and Applied Finance \textbf{22} (2019), no.~03, 1950010.

\bibitem{gerhold2019moment}
Stefan Gerhold, Christoph Gerstenecker, and Arpad Pinter, \emph{Moment explosions in the rough heston model}, Decisions in Economics and Finance \textbf{42} (2019), no.~2, 575--608.

\bibitem{gripenberg1990volterra}
Gustaf Gripenberg, Stig-Olof Londen, and Olof Staffans, \emph{Volterra integral and functional equations}, no.~34, Cambridge University Press, 1990.

\bibitem{hubalek2006variance}
Friedrich Hubalek, Jan Kallsen, and Leszek Krawczyk, \emph{Variance-optimal hedging for processes with stationary independent increments}, The Annals of Applied Probability \textbf{16} (2006), no.~2, 853--885.

\bibitem{kallsen2010variance}
Jan Kallsen and Arnd Pauwels, \emph{{Variance-optimal hedging in general affine stochastic volatility models}}, Advances in Applied Probability \textbf{42} (2010), no.~1, 83--105.

\bibitem{keller2020comparison}
Martin Keller-Ressel and Assad Majid, \emph{A comparison principle between rough and non-rough heston models—with applications to the volatility surface}, Quantitative Finance \textbf{20} (2020), no.~6, 919--933.

\bibitem{kunita1967square}
Hiroshi Kunita and Shinzo Watanabe, \emph{{On square integrable martingales}}, Nagoya Mathematical Journal \textbf{30} (1967), 209--245.

\bibitem{larsson2023chapter}
Martin Larsson, Martin Keller-Ressel, and Sergio Pulido, \emph{Chapter 5: Rough affine models}, Rough Volatility, SIAM, 2023, pp.~103--113.

\bibitem{martins2023aspects}
David Pires~Tavares Martins, \emph{Aspects of quadratic utility: mean-variance hedging in rough volatility models, and capm-type equilibria}, Ph.D. thesis, ETH Zurich, 2023.

\bibitem{neuberger1994log}
Anthony Neuberger, \emph{The log contract}, Journal of Portfolio Management \textbf{20} (1994), no.~2, 74--80.

\bibitem{nutz2023limits}
Marcel Nutz, Johannes Wiesel, and Long Zhao, \emph{Limits of semistatic trading strategies}, Mathematical Finance \textbf{33} (2023), no.~1, 185--205, DOI: 10.1111/mafi.12366.

\bibitem{rheinlander2011hedging}
Thorsten Rheinlander and Jenny Sexton, \emph{Hedging derivatives}, Advanced Series on Statistical Science \& Applied Probability, vol.~15, World Scientific Publishing Company, 2011, DOI: 10.1142/8062.

\bibitem{schweizer1984varianten}
Martin Schweizer, \emph{Varianten der black-scholes-formel}, Ph.D. thesis, Diplomarbeit ETHZ, Zurich, 1984.

\bibitem{wang2026variance}
Yingli Wang and Xiaoyu Wang, \emph{Variance-optimal hedging in the rough hawkes--heston model}, arXiv preprint arXiv:2609.08541 (2026).

\bibitem{wiman1905fundamentalsatz}
Anders Wiman, \emph{{\"U}ber den fundamentalsatz in der teorie der funktionen}, Acta Mathematica \textbf{29} (1905), no.~1, 191--201.

\end{thebibliography}

\end{document}